%% file: ms.tex
\documentclass[manuscript,screen,nonacm]{acmart}
\pdfoutput=1 

\usepackage{balance}  
\usepackage{amsthm}
\usepackage{dcolumn}
\usepackage{graphicx}
\usepackage{caption}
\usepackage{multirow}    
\usepackage{subfig}
\usepackage{multicol}
\usepackage{diagbox}
\usepackage{amsthm}
\usepackage{tikz}
\usepackage{enumerate}
\usepackage{pgfplots}
\usepackage{tikz}
\usepackage{enumerate}
\pgfplotsset{tick label style={font=\Large},compat=1.5}
\usepackage{pgfplotstable}
\usepgfplotslibrary{groupplots}
\usepackage{microtype}
\usepackage{colortbl}
\usepackage{appendix}
\usepackage{enumitem}
\usepackage{hf-tikz}
\newtheorem{myDef}{Definition}

\usepackage[linesnumbered,lined,boxed,ruled,vlined,noend]{algorithm2e}

\usepackage{mathtools}

\DeclareMathOperator*{\argmax}{arg\,max}

\newcommand{\wiki}{\texttt{wiki-Vote}}
\newcommand{\cithep}{\texttt{cit-HepPh}}

\newcommand{\NYC}{\texttt{NYC}}
\newcommand{\Tokyo}{\texttt{Tokyo}}
\newcommand{\SG}{\texttt{SafeGraph-traj}}
\newcommand{\BBC}{\texttt{Haslemere}}
\newcommand{\Italy}{\texttt{Italy}}

\AtBeginDocument{%
  \providecommand\BibTeX{{%
    \normalfont B\kern-0.5em{\scshape i\kern-0.25em b}\kern-0.8em\TeX}}}

\setcopyright{acmcopyright}
\copyrightyear{2022}
\acmYear{2022}
\acmDOI{XXXXXXX.XXXXXXX}

\begin{document}

\title{Temporal Cascade Model for Analyzing Spread in Evolving Networks with Disease Monitoring Applications}

\author{Aparajita Haldar}
\orcid{0000-0002-8485-0330}
\affiliation{%
  \institution{University of Warwick}
  \city{Coventry}
  \country{UK}
}
\email{aparajita.haldar@warwick.ac.uk}

\author{Shuang Wang}
\affiliation{%
  \institution{University of Warwick}
  \city{Coventry}
  \country{UK}
}
\email{shuang.wang.1@warwick.ac.uk}

\author{Gunduz Demirci}
\affiliation{%
  \institution{University of Warwick}
  \city{Coventry}
  \country{UK}
}
\email{gunduz.demirci@warwick.ac.uk}

\author{Joe Oakley}
\affiliation{%
  \institution{University of Warwick}
  \city{Coventry}
  \country{UK}
}
\email{j.oakley@warwick.ac.uk}

\author{Hakan Ferhatosmanoglu}
\affiliation{%
  \institution{University of Warwick}
  \city{Coventry}
  \country{UK}
}
\email{hakan.f@warwick.ac.uk}

\renewcommand{\shortauthors}{Haldar and Wang, et al.}

\input{00-abstract}




\maketitle

\input{01-intro}
\input{02-related-work}
\input{03-problem-def}
\input{04-algos}
\input{05-expt}
\input{06-conclusion}

\begin{acks}
This research is supported in part by The Alan Turing Institute under the EPSRC grant EP/N510129/1. Aparajita and Joe are supported by the Feuer International Scholarship in Artificial Intelligence.
\end{acks}

\vfill

\balance

\bibliographystyle{ACM-Reference-Format}
\bibliography{sample-base}
\appendix

\end{document}

%% file: 00-abstract.tex
\begin{abstract}
Current approaches for modeling propagation in networks (e.g., spread of disease) are unable to adequately capture temporal properties of the data such as order and duration of evolving connections or dynamic likelihoods of propagation along these connections. 
Temporal models in evolving networks are crucial in many applications that need to analyze dynamic spread. 
For example, a disease-spreading virus has varying transmissibility based on interactions between individuals occurring over time with different frequency, proximity, and venue population density.
To capture such behaviors, we first develop the \emph{Temporal Independent Cascade~(T-IC)} model and propose a novel spread function, that we prove to be submodular, with a hypergraph-based sampling strategy that efficiently utilizes dynamic propagation probabilities. 
We then introduce the notion of `reverse spread' using the proposed T-IC processes, and develop solutions to identify both sentinel/detector nodes and highly susceptible nodes. 
The proven guarantees of approximation quality enable scalable analysis of highly granular temporal networks.
Extensive experimental results on a variety of real-world datasets 
show that the proposed approach significantly outperforms the alternatives in modeling both \textit{if} and \textit{how} spread 
occurs, by considering evolving network topology as well as granular contact/interaction information.
Our approach has numerous applications, including its utility for the vital challenge of monitoring disease spread. 
Utilizing the proposed methods and T-IC, we analyze the impact of various intervention strategies over real spatio-temporal contact networks. 
Our approach is shown also to be highly effective in quantifying the importance of superspreaders, designing targeted restrictions for controlling spread, and backward contact tracing. 


\end{abstract}

%% file: 01-intro.tex
\section{Introduction}
\label{section:1}
Research on modeling cascades in networks traditionally focused on identifying the influential nodes in social networks~\cite{Domingos:2001}. 
In this paper, we introduce a temporal cascade model for analyzing `reverse spread' and `expected spread' in evolving networks to identify i) \emph{sentinel nodes} and ii) \emph{susceptible nodes}, respectively. In contrast to traditional approaches which identify influential nodes that spread activation, we aim to identify 
nodes that i) quickly detect activation anywhere in the network (sentinel nodes) or ii) easily collect activation from anywhere in the network (susceptible nodes). The proposed setting has many applications, ranging from 
identifying targets susceptible to influence, to detecting disease outbreaks and analyzing intervention strategies using contact networks. 

We aim to address several challenges that existing cascade models do not handle effectively: 
i) \textbf{Granular data} -
The cascade model should incorporate highly granular data and additional features that are available. For example, information about individual contacts, rather than relying only on aggregated mobility patterns, and about venue type and popularity are crucial for COVID-19 monitoring and the design of intervention strategies. While there is growing interest in generating contact networks using mobile data,
there is no model that is designed to utilize individual-level spatio-temporal data.
Recent data-driven solutions, despite being designed specifically for disease monitoring~\cite{chang2021mobility,wang2021heterogeneous}, cannot handle individual-level trajectory data, and contact tracing is considered only in an aggregated form (e.g., total population changes in sub-districts) in these models.
ii) \textbf{Evolving connectivity patterns} -
Addition/deletion of edges and varying contact frequencies, e.g., about interactions between infected individuals over time, must be captured.
Rather than producing solutions from partially observed data that are more likely to be sub-optimal~\cite{ohsaka2016dynamic}, there is a need for a solution set that is optimized across time-varying contact events over an entire time window. 
iii) \textbf{Dynamic propagation patterns} -
The various spatial and temporal characteristics above need to be efficiently incorporated for simulating the propagation and designing interventions.  
There has been no solution with provable quality guarantees to handle dynamic propagation at different rates, i.e., where the likelihood of the spread of activation along connections can exhibit arbitrary temporally evolving patterns. 
iv) \textbf{Approximation guarantees} - Prior independent cascade (IC) based solutions lose approximation guarantees when modified to support time windows or graph snapshots. To allow for scalable analysis that can efficiently support the use of highly fine-grained, temporally evolving graph data, rigorous approximation guarantees must be maintained in the temporal spread modeling algorithm.

To address the above challenges, 
we introduce the \emph{Temporal Independent Cascade~(T-IC)} model for detection of outbreaks, and propose solutions for novel analytics tasks with approximation guarantees within $1\!-\!\frac{1}{e}$ of optimal. The model includes a novel spread function, which we prove to be submodular, that utilizes dynamic propagation probabilities for every edge in the network. We introduce two distinct objectives and associated T-IC based solutions in this context: i) finding sentinel nodes, and ii) finding susceptible nodes. We illustrate the application of our approach to disease monitoring in evolving networks to i) detect \emph{if} there is an outbreak by selecting (testing) a limited number of nodes (individuals), and ii) understand \textit{how} it is likely to spread by identifying the most susceptible nodes (individuals for contact tracing). We further analyze a range of intervention strategies to contain an outbreak using the T-IC model and an application-driven propagation probability function derived from real disease spread characteristics. This approach is showcased on real-world  location-based networks and pandemic datasets.
We also present our results on several other types of datasets, and note that our model can be easily applied to any highly granular data that becomes available in future.


Our approach involves three stages. We first formalize the evolving network where 
edge connections may be added/deleted and 
the propagation probability between two nodes can 
vary over time. A timely example for this is a spatio-temporal contact network, e.g., in contact tracing applications to monitor the spread of a disease with a varying likelihood of transmission that is based on contact duration, proximity, population density, and other customizable factors determined by experts. For such applications, we construct the network from mobile locations and assign the edge probabilities based on the ``force of infection'' for the corresponding contagious disease~\cite{kissler2020sparking,mittal2020mathematical}.  

Second, the \emph{Temporal Independent Cascade~(T-IC)} model is defined over the network, which handles dynamic propagation rates and allows an active node to repeatedly try to activate its neighbors. 
This is an enhancement of the existing popular IC model, that takes temporal order of activations into account and allows dynamic changes in both propagation probability and network connectivity patterns.


Third, we define novel spread functions to address the objectives of finding sentinel nodes and finding susceptible nodes. 
We propose an efficient hypergraph-based sampling strategy to capture evolving connections and dynamic propagation rates in the network. We then use T-IC to generate `random reachable sets' that can reflect the patterns of spread within this evolving network. 
These random reachable sets are used to develop two solutions, \emph{Reverse Spread Maximization (RSM)} and \emph{Expected Spread Maximization (ESM)}, that we use to identify sentinel and susceptible nodes respectively, with provable approximation guarantees for the optimality of these solution sets.
We employ RSM to find sentinel nodes, i.e., a set of nodes where at least one is likely to be activated regardless of where the spread begins. 
RSM is useful to detect \emph{if} there is spread within a population. 
An application of this approach is ``sentinel surveillance'' and the early detection of a disease outbreak by using sensors at the set of sentinels~\cite{holme2017three,christakis2010social,bajardi2012optimizing}.
To understand \textit{how} the spread is likely to take place, 
we also study ESM 
to identify susceptible nodes, i.e., a set of nodes that accumulate the most spread from arbitrary seeds.
In a disease contact network, these individuals would be the first choice for forward and backward contact tracing (to identify others who either may have been infected by these individuals or may have passed along infection to these individuals), being at high risk of catching as well as transmitting infection. It is important to note that the traditional approaches of merely ranking the nodes by importance or ability to spread activation (as with influence maximization) 
are insufficient for our proposed setting. 

The main technical contributions of the paper are as follows:
\vspace{-1mm}
\begin{itemize}[leftmargin=*]
    \item A temporal independent cascade model (T-IC) for networks with dynamic propagation rates is introduced, along with a `reverse spread' function which is submodular and provides a formal approximation guarantee on solutions.
    \item A sampling strategy based on hypergraph construction is proposed to handle evolving connections in the network, and dynamic spread patterns are reflected in `random reachable sets' that form hyperedges. 
    \item We propose RSM for selecting sentinel nodes (e.g., individuals to test with priority) to detect \emph{if} there is any spread, and ESM for identifying susceptible nodes (e.g., individuals to contact trace) to capture \emph{how} the spread can occur.
    
    
    \item Experiments using highly granular real-world location, contact, and social network datasets confirm that the proposed solutions are effective in identifying i) a sentinel set with highest `reverse spread' coverage and success of detecting spread, and ii) a susceptible set with highest likelihood of activated nodes included based on the `expected spread' pattern. Further results are presented for disease monitoring applications, including mitigation and intervention strategies such as targeted restrictions and backward contact tracing.
    

\end{itemize}


%% file: 02-related-work.tex
\vspace{-1mm}
\section{Related Work}
\label{section:2}

While there is no prior work on capturing sentinel and susceptible sets in temporal networks as studied in this paper, there is extensive work in the 
areas of information diffusion and influence maximization, especially for social networks~\cite{Domingos:2001,Leskovec:2007,Chen:2009, Chen:2010, Wang:2010,guille2013information, tang2015influence}.  Under the widely used propagation models, Independent Cascade~(IC) and Linear Threshold, finding a subset of users that maximizes the expected spread is shown to be NP-Hard~\cite{Kempe:2003}.  
There are studies using time constraints within IC, which use a constant probability of spread or ignore repeated activations from an active node~\cite{Chen:2012,Liu:2012,kim2014ct}. In contrast, our approach exposes temporal factors in the continuous activation process and allows customizable formulations of propagation rates. 
Prior work on evolving networks typically focuses on maximizing influence~\cite{gayraud2015diffusion,yang2017tracking,ohsaka2016dynamic,han2018efficient, tang2018online} as opposed to the objectives that we study, for which we also provide approximation algorithms with quality guarantees. 

Several recent studies on data-driven disease monitoring~\cite{wang2018predicting,benzell2020rationing,chang2021mobility,wang2021heterogeneous} use the same types of location-based social network datasets, e.g., Foursquare~\cite{yang2014modeling} and SafeGraph~\cite{Safegraph2020}, that we also cover in our experiments. It is worth noting that studies such as these do not attempt to perform individual-level trajectory analysis using contact tracing data. Such location-based data is publicly available only in an aggregated form (e.g., total population changes in sub-districts), or is used in a decentralized fashion to determine infection risk for individuals~\cite{jiang2022survey}. Our approach enables a more fine-grained analysis in any specific propagation scenarios (e.g., with varying population density and varying proximity between interacting individuals), as illustrated in our experiments using data that have been collected for analyzing pandemic spread, e.g., the data from the BBC Pandemic app~\cite{kissler2020sparking} and Italy mobility data captured during the COVID-19 outbreak~\cite{pepe2020covid}. We also utilize Foursquare trajectories and generate granular trajectories from the SafeGraph data, in lieu of their aggregated versions, for a more granular approach to spread modeling. 

There is extensive epidemiological research on disease modeling starting from the SIR framework~\cite{kermack1927contribution} and its extensions, where differential equations govern the transition rates between Susceptible, Infectious, and Recovered stages. 
This framework has led to numerous spread models being designed that handle datasets at a low-granularity aggregated level (e.g., district-level infection case counts).
One such approach deals with greedy immunization strategies restricted to local behavior such as node degrees~\cite{prakash2010virus}. 
SIR is also used to study the impacts of static and temporal network structures on outbreak size, sentinel surveillance, and vaccination objectives~\cite{holme2017three,holme2018objective}. 
An extension of the traditional SIR approach, SEIR, has been studied in the context of COVID-19~\cite{chang2021mobility}, with SafeGraph sub-district-level spread modeling performed to aid policy making around the partial easing of lockdowns and other such considerations. Results are calibrated against real New York Times COVID-19 aggregated case counts. However, individual-level trajectory analysis is not addressed.

All of the aforementioned equation-based compartmental models in the domain use aggregate data and assume homogeneous mixing within the population without considering temporally ordered meeting events. Instead, our solution offers highly granular and efficient 
predictive models, with novel goals of identifying critical subsets (i.e. sentinel and susceptible nodes) in the contact network. 
A recent Hawkes-process based variation of SIR has been proposed by performing agent-based infection simulations to assign COVID-19 risk scores to individuals~\cite{rambhatla2022toward}. 
Our approach has orthogonal perspectives where we use our assigned transmission risk scores to determine optimal solution sets that are either sentinels or susceptible individuals.
T-IC is also able to utilize individual-level location sequences 
to trace the disease spread spatially,  
resulting in support for more use cases such as ``backward contact tracing'' (using knowledge of active nodes to trace initial seeds), which has attracted attention in the context of COVID-19~\cite{endo2020implication}. 
The need for dynamic network analysis for forward and backward influence tracking has been highlighted in the literature~\cite{aggarwal2012influential}.

%% file: 03-problem-def.tex
\section{Problem Definition}
\label{section:3}

We now present the Temporal Independent Cascade (T-IC) model, and define our optimization objectives, for tracing the spread over an evolving network with dynamic propagation probabilities.
\vspace{-2mm}
\subsection{Temporal network}

In a standard IC model, information flows/diffuses/propagates through the network via a series of cascades. Nodes may either be active (already influenced by the information that is propagating through the network) or inactive (either unaware of the information that is propagating but has not reached it by this point, or not influenced during the propagation that did reach it).
The standard IC model assumes a static probability distribution over a static graph structure. This IC process is simulated over a graph $G\!=\!(V,E,p)$ where each edge $(u,v) \in E$ is associated with a constant probability function $p \colon E \mapsto [0,1]$, reflecting the likelihood of activation when nodes $u\!\in\!V$ and $v\!\in\!V$ have a common edge (e.g., a common meeting point in the location histories of two individuals). Propagation starts from an initial seed set 
in $V$ (the only nodes active 
at step $0$). 
Propagation takes place in discrete steps with each active node $u$ during step $i$ being given a single chance to activate its currently inactive neighbor $v$ with some probability $p(u,v)$.
That is, at every step $i\!\geq\!1$, any node made active in step $i\!-\!1$ has a single chance to activate any one of its inactive neighbors. The process continues with nodes remaining active once activated, until no further propagation is possible.
Therefore, this is a stochastic process that requires a large number of simulations to accurately determine the spread of information.

Since the standard IC model uses a static probability distribution over a static network, it is insufficient to handle evolving graphs with changing propagation rates.
In our setting, both the structure of $G$ and the propagation rates can change dynamically. 
For example, a disease spread model needs to consider the order and duration of interactions within a population, and the varying infectivity of a virus over its lifetime. Therefore, we define a new temporal IC model for a temporal network.

\begin{myDef}[Temporal Network]
Given a discrete time domain $T\!=\!\{1,2,...,n\}$, a temporal network is a graph $G = (V, E, p^t)$ where for each time interval $t\!\in\!T$, edge set $E$ is associated with a different propagation probability distribution $p^{t}\!\colon\!E \mapsto\![0,1]$. That is, each edge $(u,v)\!\in\!E$ has $n$ probabilities $p^t(u,v)$, one for each $t\!\in\!T$.
\label{definition:interaction-network}
\end{myDef}


Since $u$ and $v$ may be linked multiple times in $T$, the corresponding $p^t(u,v)$ for every interval $t$ needs to be separately maintained. An evolving graph is represented by adding all edges and assigning $p^t(u,v)=0$ when there is no $(u,v)$ connection in interval $t$. Moreover, a rigorous formulation for $p^t(u,v)$ is needed to describe more complicated cases of spread, which we discuss in Section~\ref{probability_assignment}. 

\subsection{Temporal Independent Cascade model}
To model the spread of activation in the temporal graph $G$, we introduce the Temporal Independent Cascade (T-IC) model as an enhancement 
of the popularly used standard IC model for this novel setting.
Given time intervals $i,j \in T$ such that $i\!<\!j$, the T-IC model proceeds from $i$ to $j$ as follows: Let $A_t$ denote the set of initially active nodes at the beginning of time interval $t$. Within each interval~$t\!\in\![i,j]$, the standard IC model is executed once under probability distribution $p^t$ on edges and proceeds until no further spread is possible. The set of all activated nodes at the end of interval $t$ is $A_{t+1}$, which thus also represents the active nodes at the beginning of the next time interval $t\!+\!1$ (active nodes remain active in subsequent intervals). 
This process is executed for each interval $i,i\!+\!1,\ldots,j$ and the final set of active nodes $A_{j+1}$ is obtained after interval $j$.

In other words, we do not modify the standard IC process, but instead run it to completion independently within each discrete time interval $t\!\in\![i,j]$ under the corresponding probability distribution defined over edges. For the entire chosen time window ($[i,j] \in T$), stacking several of these IC processes, and treating activated nodes as the seeds for the next run, allows us to mimic the spread over an evolving graph without sacrificing the approximation guarantees.
Furthermore, if a node is activated in a specific time interval, it can continue to activate its neighbors in subsequent time intervals. 
For example, a person infected with a virus may continue to spread infection during interactions with people for as long as they are contagious (e.g., $14$ days for COVID-19 \cite{bazant2020guideline}). 
This continuous activation during the T-IC process, along with the propagation probability formulations, makes it possible to reflect real-world spreading phenomena (e.g., for disease monitoring) with our model. 


\begin{figure}[t]
\centering
\subfloat[Propagation probability increases with population density and duration]{\includegraphics[width=0.75\linewidth]{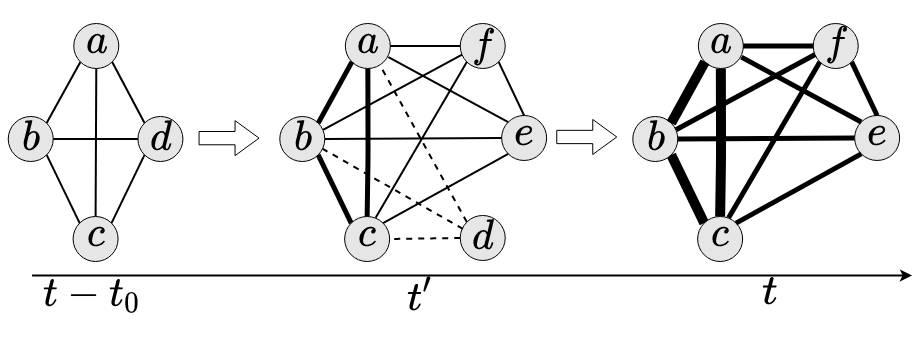}\label{fig:example-t-ic-1}}
\\
\subfloat[Propagation probability increases with proximity and duration]{\includegraphics[width=0.75\linewidth]{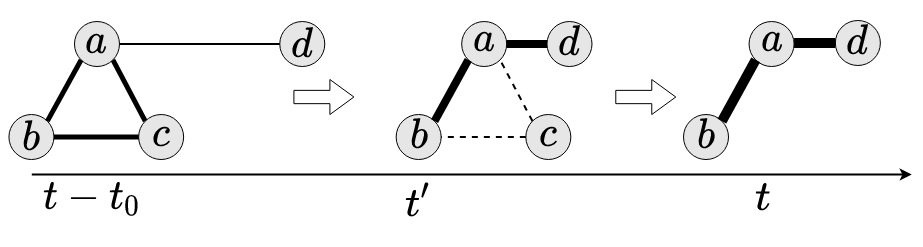}\label{fig:example-t-ic-2}}
\caption{Dynamic propagation probabilities in the network reflect temporal characteristics. Edge connection indicates proximity, with thicker edge for higher probability (based on population density, proximity, and duration of contact) and dotted edge for latent probability after deletion.}
\label{fig:example-t-ic}
\end{figure}

Figure~\ref{fig:example-t-ic} illustrates the impact of temporal order and dynamic connections in the spread model. 
Suppose node $c$ is initially active. 
In the first case (Figure~\ref{fig:example-t-ic}\subref{fig:example-t-ic-1}), nodes $a$ and $b$ have higher likelihood of activation than nodes that arrive later, such as $e$ or $f$, or nodes that leave, such as $d$. The greater chance of activation due to the prolonged duration of contact with the active node $c$, as well as the increased risk of activation due to the high density of nodes, is reflected in the thicker edge connections to $a$ and $b$. 
In the second case (Figure~\ref{fig:example-t-ic}\subref{fig:example-t-ic-2}), $a$ and $b$ are again more susceptible due to their proximity to $c$, because $c$ leaves before $d$ approaches closer. The risk of activation for $d$ then increases as the node approaches closer to the other nodes that may have been activated by $c$.
We develop a propagation probability assignment that can capture all these factors for different application settings, which we describe in Section~\ref{probability_assignment}.


\subsection{Optimization objectives}

Our first objective is to identify the sentinel nodes, i.e., the set of nodes that maximizes the probability of detecting any activations in a network where the set of initially active nodes is not known. A practical application is 
to detect an outbreak 
using minimal resources (e.g., medical tests).
Testing a targeted set of individuals can be an efficient way to detect the presence of disease within a population before it is widespread, akin to the concept of ``sentinel surveillance'' in network epidemiology~\cite{holme2017three}. Therefore, our optimization objectives focus on finding $k$-element subsets of sentinel nodes. The proposed solution can also 
be run until termination over an indefinite time window, but we note that when there is a temporal constraint within which to detect outbreaks, a truncated solution set (choosing only $k$ nodes) can sufficiently cover the entire relevant spread.

The objective is to maximize the probability that at least one node in a $k$-element solution set $S$ becomes active after a random T-IC process (i.e., one that starts with randomly selected seeds) within a given time window $[i,j]$. 
We note that the `temporal reverse spread maximization' objective, as defined below, corresponds to this goal of identifying a $k$-element set of sentinel nodes.
Optimizing the success rate of detection of spread anywhere in the network by using sentinels can be achieved by maximizing the expected amount of `reverse spread' $\phi(\cdot)$. Expected reverse spread can be defined as the expected number of nodes that can spread activation to the nodes in $S$. 
Therefore the expected reverse spread of set $S$ on $G$ within $[i,j]$, denoted by $\phi_{ij}(G,S)$, is the expected number of nodes that can activate at least one node in set $S$ during a random T-IC process in $[i,j]$. We discuss how to compute $\phi(\cdot)$ in Section~\ref{section:4}.
The problem of maximizing the expected reverse spread $\phi_{ij}(G,S)$ can be formally defined as follows:

\begin{myDef}[Temporal Reverse Spread Maximization]
Find the $k$-element subset of nodes $S^* \subset V$ such that
\begin{equation}
S^* = \argmax_{S \subset V, |S| = k} \phi_{ij}(G,S)
\end{equation}
\label{definition:rev-max}
\end{myDef}




Definition~\ref{definition:rev-max} addresses the problem of 
detecting \emph{if} there is any spread (e.g., a disease outbreak).
To understand \textit{how} the spread takes place within the network (e.g., in order to contact trace), we need to identify high priority nodes that collect activation, i.e., nodes that are the most susceptible to activation. 
For example, in disease monitoring applications, the objective is to identify the subset $S$ for testing who are all highly likely to be infected. This is because backward contact tracing~\cite{endo2020implication,aggarwal2012influential} from these nodes may offer important candidates for immunization to prevent spread of infection to nodes of $S$. 

Hence, we next aim to identify the $k$-element subset $S$ 
containing the maximum expected number of active nodes after T-IC process in $[i,j]$. 
Similarly to the temporal reverse spread maximization objective, the `temporal expected spread maximization' objective defined below corresponds to the goal of identifying a $k$-element set of susceptible nodes. 
This second problem is formally defined as follows:

\vspace{-2mm}
\begin{myDef}[Temporal Expected Spread Maximization]
Let $I(s)$ be an indicator random variable for node $s\!\in\!S$ such that 
\begin{equation}
I(s) = \begin{cases}
        1 &\text{if $s$ is activated} \\
        0 &\text{if $s$ is not activated}
\end{cases}
\end{equation}
\noindent after executing random T-IC process.
Find a $k$-element subset of nodes $S^*\!\subset\!V$ that maximizes the expected number of active nodes in $S^*$ 
\begin{equation}
S^* = \argmax_{S \subset V, |S| = k} \mathbb{E}\left[\sum\limits_{s\in S}I(s)\right]
\end{equation}
\label{definition:max-expected-spread}
\end{myDef}

In the next section, we describe our proposed solutions to address the above objectives, namely \emph{Reverse Spread Maximization} (RSM) and \emph{Expected Spread Maximization} (ESM).

%% file: 04-algos.tex
\section{Reachable Set Sampling Based Algorithm}
\label{section:4}

While other models lose the rigorous approximation guarantees with changes to support temporal spread, we introduce a novel significant contribution to the well-defined sampling approach that allows us to preserve optimality guarantees in a new dynamic setting. 
We first show that the defined 
reverse spread function $\phi(\cdot)$ is submodular in Theorem~\ref{theorem:submodular}. It follows that the standard hill-climbing greedy algorithm achieves $1\!-\!\frac{1}{e}$-approximation guarantee, i.e., a solution that uses this function can be approximated to within $1\!-\!\frac{1}{e}$ of optimal~\cite{Kempe:2003}.

\begin{theorem}
Under the T-IC model, function $\phi_{ij}(\cdot)$ is submodular.
\label{theorem:submodular}
\end{theorem}
\begin{proof}
Let $G=(V,E,p)$ be a directed graph where each edge $(u,v)\!\in\!E$ is associated with a weight $p(u,v)$ denoting the probability that spread occurs from $u$ to $v$. 
\citet{Kempe:2003} showed that the IC model induces a distribution over graph $G$, such that a directed graph $g=(V,E^\prime)$ can be generated from $G$ by independently realizing each edge $(u,v)\!\in\!E$ with probability $p(u,v)$ in $E^\prime$.
In a realized graph $g\!\sim\!G$, nodes reachable by a directed path from a node $u$ are its reachable set $R(u,g)$, and correspond to the nodes activated in one instance of the IC process with $u$ as the initially active seed node.
They proved that for $S\subset V$, the spread function $\sigma(S,g)=|\cup_{u\in S}R(u,g)|$ is submodular.

Similarly, the T-IC model induces a distribution over $G\!=\!(V,E,p^t)$, where the IC model is executed independently in each discrete time interval $t\!\in\![i,j]$ under the corresponding probability distribution defined over edges.
Additionally, an activated node remains active in subsequent intervals, getting multiple chances to activate its neighbors.
So, directed graph $g_{ij}=(V,E^\prime)$ can be generated as follows:
For intervals $t\!=\!i,i\!+\!1,\ldots,j$, each edge $(u,v)\!\in\!E$ is realized in $E^\prime$ with probability $p^t(u,v)$, only if node $u$ is active at the beginning of interval $t$. 
Hence, the reachable set $R(s,g_{ij})$ corresponding to a node $s$ on the generated graph $g_{ij}$ consists of all the nodes that are reachable and activated by time interval $j$ by the seed $s$ that was initially active in time interval $i$.

Let $g_{ij}^T$ denote the transpose of $g_{ij}$, obtained by reversing all its directed edges.
Reachable set $R(r,g_{ij}^T)$ corresponds to all seed nodes that, if active in interval $i$, would have the ability to activate the receiving node $r$ by time interval $j$.
Given a set of nodes $S$, let the reverse spread $\phi(S, g_{ij})$ denote the number of nodes that can reach some node in $S$. That is, $\phi(S, g_{ij})=|\bigcup_{u\in S}R(u,g_{ij}^T)|$.
Since $\phi(S,g_{ij})\!=\!\sigma(S,g_{ij}^T)$, the submodularity of $\phi(S,g_{ij})$ follows.
Therefore, the expected reverse spread $\phi_{ij}(G,S)=E(\phi(S,g_{ij}))$ is submodular, being a linear combination of submodular functions.
\end{proof}

\citet{borgs2014maximizing} use a state-of-the-art sampling strategy to build a hypergraph representation and estimate the spread of activation. We enhance this technique to handle dynamic propagation rates and identify solution sets for both our defined tasks (i.e., identifying sentinel nodes and susceptible nodes).
Our algorithm and sampling strategy use a novel process of generating the hypergraph 
to encode the reverse spread of any given subset of nodes via its nets. A hypergraph is a generalization of a graph in which two/more nodes (pins) may be connected by a hyperedge (net). 
The two-step sampling strategy is as follows: i) we execute random T-IC processes (that start with random active seeds) on the temporal network, and ii) for each execution of a T-IC process, we construct a net whose pins are the nodes that are activated during the process.

As shown in Theorem~\ref{theorem:submodular}, $g_{ij}$ can be drawn from the distribution induced by a T-IC model on $G$.
The edges in the graph $g_{ij}$ are tried to be realized by traversing only live edges (i.e., edges where the starting node is already active). Due to this constraint of only considering live edges, this enforces a dynamic nature to the spread as it takes place by respecting the temporal ordering of connections.
If a node $v$ is reachable from many different nodes in $g_{ij}$, then it is more likely that this node will be activated 
by time interval $j$. Since any random seed in time interval $i$ is equally likely to start the spread, the existence of more paths that lead to the node $v$ results in a higher likelihood of its activation. 
This means that the reachable set of nodes $R(u,g)$ (i.e., all the nodes from the realized graph $g_{ij}$ that are reachable by a directed path of edges from the node $u$), which depends on the random seed node $u$, is one among many possible sets of activated nodes at the end of a random T-IC process on the randomly sampled $g_{ij}$.

Hence, the solution depends on two levels of randomness that are encountered during the hypergraph construction: i) the sampling strategy for $g_{ij}\!\sim\!G$, and ii) the computation of $R(u,g_{ij})$ given a random seed $u$. 
The former depends on the probability distribution induced by the T-IC model over $G$, while the latter depends on the seed node $u$.
We refer to such a reachable set that is generated by two levels of randomness as a `random reachable set' $RR(u,g_{ij})$. 
In disease modeling applications, most outbreaks are thought to start with one infected person~\cite{kiss2017mathematics, hethcote2000mathematics}. Therefore, we consider a single seed node in our experiments on all datasets.

The main sampling step is repeatedly performed to build a hypergraph $H\!=\!(V,N)$ where each net 
$n_u\!\in\!N$ is independently generated by executing a random T-IC process from seed $u$. The hypergraph corresponds to a random reachable set $RR(u,g_{ij})$, i.e.,
$\operatorname{pins}(n_u)=RR(u,g_{ij})$.
The solution quality and concentration bounds thus depend on the number of nets generated to build the hypergraph \cite{borgs2014maximizing}.

\input{plots/rsm_solution_algorithm}

Note that $H$ and $G$ are composed of the same set of nodes $V$. 
For a solution set $S$, the number of nodes sharing a net with at least one node in set $S$~(which we refer to as $\operatorname{deg}(S)$ henceforth) corresponds to the number of times a node in $S$ gets activated during the random T-IC processes executed to compute the random reachable sets. To select $S$ as a collection of sentinel nodes, higher $\operatorname{deg}(S)$ will be more likely to detect spread in the network, which can be understood as follows:
The degree of a node 
in the hypergraph is the sum of $|N|$ Bernoulli random variables~\cite{borgs2014maximizing}. This is because the inclusion of a node $v$ in a random reachable set $RR(u,g_{ij})$ and in $\operatorname{pins}(n_u)$ can be considered as a Bernoulli trial with success probability $p_v$, where $p_v$ denotes the probability that $v$ gets activated in a random T-IC process. That is, the hypergraph node degrees are binomially distributed with an expected value of $p_v\!\times\!|N|$. This implies that $p_v\!=\!\mathbb{E}[deg(v)/|N|]$.
Therefore, this node degree corresponds to the estimation of reverse spread of node $v$, since the reverse spread can be written as $\phi_{ij}(v)\!=\!|V|\!\times\!p_v$.
Similar to the node degrees in hypergraph $H$, the expected value $\mathbb{E}[deg(S)/|N|]$ corresponds to the probability that \emph{at least one} node in $S$ gets activated during a random T-IC process.
Therefore, the degree of a set $S$ of nodes in hypergraph $H$, corresponds to the reverse spread $\phi_{ij}(S)\!=\!|V|\!\times\!\mathbb{E}[deg(S)/|N|]$, which can be estimated well if a sufficient number of nets are built.

We next describe two algorithms, RSM and ESM, to efficiently compute solution sets for the two tasks corresponding to Definitions~\ref{definition:rev-max} and~\ref{definition:max-expected-spread} respectively.


\subsection{Reverse Spread Maximization solution}
In the hypergraph $H\!=\!(V,N)$, if a node connects many nets~(i.e., its degree is high), then that node 
has a high probability of being activated during a random T-IC process. 
Similarly, if a set $S$ of nodes covers 
many of the nets~(random reachable sets), then its expected reverse spread $\phi_{ij}(G, S)$ is likely to be higher. In other words, there is a larger set of nodes that all have a chance to activate at least one node of $S$ within the time window $[i,j]$.

As in the maximum coverage problem, we want to cover the maximum number of nets~(elements) in the hypergraph $H$ by choosing a solution set $S$ of $k$ nodes~(subsets). This step is therefore equivalent to the well-known NP-Hard maximum coverage problem~\cite{vazirani2013approximation}.
\citet{borgs2014maximizing} show that the maximum set cover computed by the greedy algorithm on the hypergraph yields $(1\!-\!\frac{1}{e}\!-\!\epsilon)$-approximation guarantee for the influence maximization problem.
Here, the parameter $\epsilon$ relates the approximation guarantee to the running time of the algorithm, and the solution quality 
improves with the increasing number of nets in the hypergraph.

Algorithm~\ref{algorithm:rm} displays the overall execution of the proposed solution.
It generates a number of random reachable sets by first drawing a graph $g_{ij}$ from the distribution induced by T-IC model on the input graph $G$ and then performing a breadth-first search~(BFS) starting from a randomly selected node $u$.
This randomized BFS through time intervals proceeds such that the set of source nodes at each interval are the activated nodes in the preceding interval (lines $4$--$11$).
Thus each edge $(u,v)\!\in\!E$ is searched with probability $p^t(u,v)$ in time interval $t$.
All nodes activated during a random BFS form a random reachable set and are connected by a net in hypergraph $H$~(line $12$). 
After generating the hypergraph $H$ with $|N|$ nets, the algorithm repeatedly chooses the highest degree node at each iteration, adds it to the solution set, and subtracts this node together with all incident nets from the hypergraph.
This is done repeatedly until a $k$-element subset of nodes, which is the resulting solution set $S$, is computed~(lines $13$--$15$). This algorithm generates a solution of sentinel nodes for Definition~\ref{definition:rev-max}.

\subsection{Expected Spread Maximization solution}
In order to maximize the expected number of active nodes in $S$, all the nodes having the highest probability of being activated should be included in the solution set, since the expected number can be given as
$\mathbb{E}\left[\sum_{s\in S}I(s)\right] = \sum_{_{s\in S}} p_s$.
\noindent
Hence, the problem in Definition~\ref{definition:max-expected-spread} can be solved by a modified version of Algorithm~\ref{algorithm:rm}. 
The final step 
(lines $13$--$15$) is replaced by 
selecting the top-$k$ nodes having the most incident nets. That is, we sort the nodes in descending order of degrees in $H$ and choose 
the first $k$ nodes. 

\begin{figure}[tb]
\includegraphics[width=0.55\textwidth]{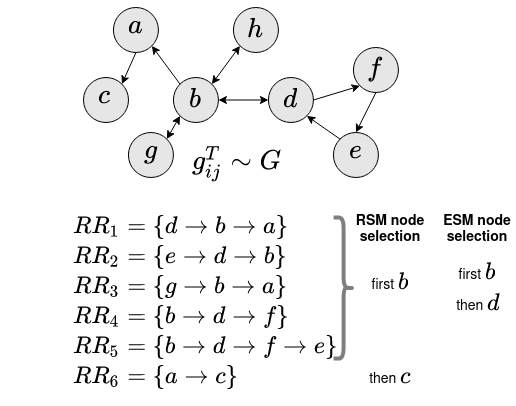}
\centering
\caption{RSM vs ESM on $g_{ij}^T$.}
\label{fig:example-rsm-esm-1}
\end{figure}

\begin{figure}[tb]
\includegraphics[width=0.55\textwidth]{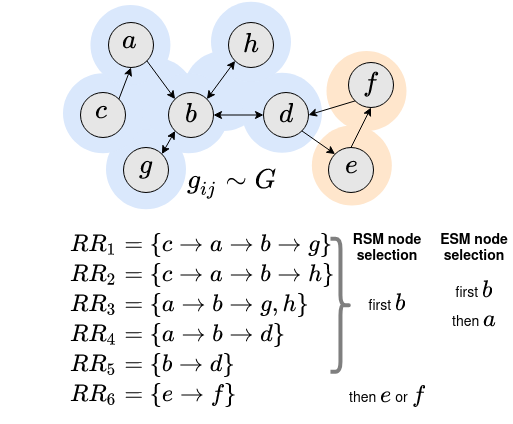}
\centering
\caption{RSM vs ESM on $g_{ij}$.}
\label{fig:example-rsm-esm-2}
\end{figure}

RSM node selection on a transposed graph $g_{ij}^T$ helps to identify the most influential nodes for influence maximization solutions, since reversing the edges makes these nodes the best spreaders in the original graph (Figure~\ref{fig:example-rsm-esm-1}). However, as illustrated in Figure~\ref{fig:example-rsm-esm-2}, we apply our sampling strategy using $g_{ij}$ directly, allowing solutions for our novel objectives of finding sentinel/susceptible nodes. Clearly, the sampling step is dependent on temporal dynamics which change upon transposing $g_{ij}$, therefore solutions for our novel objectives are distinct from influence maximization approaches in literature.
Specifically, when applied to disease modeling, the RSM approach can detect an outbreak efficiently, e.g., sentinel nodes $b$ and $e$ (or $f$) collectively correspond to greater coverage of the network than that offered by the ESM solution (nodes $b$ and $a$). The ESM solution appears in RR-sets more frequently, and are nodes that are all highly likely to be infected and therefore important for contact tracing efforts.
Identifying susceptible nodes in a social network also has interesting use cases such as reducing/capturing rumor spread more effectively. 




%% file: plots/rsm_solution_algorithm.tex
\begin{algorithm}
\SetKwInOut{Input}{input}
\Input{$G=(V,E,p^t)$, $i$, $j$, $K$, $|N|$}
\SetAlgoLined
$H=(V,N=\emptyset)$; $S=\emptyset$\\
\For{$n = 1$ \textbf{to} $|N|$}{
    Select source node $s \in V$ uniformly at random; $A = \{s\}$\\
    \For{$t = i$ \textbf{to} $j$}{
        $BFS\_Q = A$ \\
        \While{$BFS\_Q \neq \emptyset$}{
        $u = \mathrm{dequeue}(BFS\_Q)$\\
        \ForEach{$(u,v) \in E$}{
            Draw $p \in [0,1]$ uniformly at random\\
            \If{$p \leq p^t(u,v)$ \textbf{and} $v \not\in A$ }{
               $A = A \cup \{v\}$; $\mathrm{enqueue}(BFS\_Q,v)$\\
            }
        }
        }
    } 
    $N=N\cup \{A\}$\\
}
\For{$k = 1$ \textbf{to} $K$}{
$v_k = \argmax_v \mathrm{deg}_{H}(v)$; $S = S \cup \{v_k\}$\\
Remove $v_k$ and all of its incident nets from $H$\\
}
\Return S 
\caption{RSM solution}
\label{algorithm:rm}
\end{algorithm}

%% file: 05-expt.tex
\section{Experimental Evaluation}
\label{section:5}

For each real dataset in our experiments, we build a temporal network on which we execute the T-IC process, and 
evaluate the solutions in terms of identifying sentinel nodes and susceptible nodes as defined in Section~\ref{section:3}. Experiments are executed on an Ubuntu 20.04 machine with 16 Intel 3.90 GHz CPUs and 503 GB RAM.
The code and data used in our experiments are publicly available at: \href{https://github.com/publiccoderepo/T-IC-model}{https://github.com/publiccoderepo/T-IC-model}




\subsection{Experimental Settings}
\subsubsection{Datasets:\\}
\noindent
We used seven real datasets in our experiments. 
For the use case of disease monitoring, we use five 
location-based networks. The first two are spatio-temporal network datasets that we construct using the locations (check-ins) of Foursquare users, in line with other research studies on disease monitoring applications~\cite{wang2018predicting,benzell2020rationing}. We refer to them as \textbf{{\NYC}} and \textbf{{\Tokyo}} datasets. 
The {\NYC} and {\Tokyo} datasets~\cite{yang2014modeling} record check-in times, (anonymized) user IDs, venue IDs, venue locations, and venue categories. 
The temporal and spatial information from these are used to build edge connections in the network of users, selecting $25$ consecutive days' data. To alleviate sparsity, the nodes for all users visiting the same venue in the same day are connected bidirectionally.


Our third location dataset is based on SafeGraph~\cite{Safegraph2020}, which has been used to analyze mobility patterns for COVID-19 mitigation~\cite{chang2021mobility}. SafeGraph contains POIs, category information, opening times, as well as aggregate mobility patterns such as the number of visits by day/hour and duration of visits.  Using these mobility patterns, we generate synthetic trajectories for $2K$ individuals visiting $100$ unique POIs in the NYC area over $25$ days. 
To build an individual's trajectory, for each day of the week for three consecutive weeks, 
we select and assign sequential visit locations to appropriate timestamps (based on travel time and visit duration) as follows:
i) each individual receives a random start timestamp for travel, a random start POI location, and a random trajectory length that determines the number of POIs to visit, ii) SafeGraph dwell time estimates and a random distance-based travel time are used to determine the timestamp for reaching the next location, iii) depending on this timestamp, POIs are filtered out from the candidate list (based on opening time, category information, and distance from current location) to ensure that the trajectory sequences generated are feasible and realistic, and iv) the next location POI is selected from among the remaining candidates, and the process (steps ii--iv) is repeated until the full length trajectory is complete (where no candidates exist, the trajectory is truncated).
We then construct the corresponding contact network by connecting (bidirectionally) nodes that appear in the same location at the same time, considering 5 minute intervals to determine this overlap. We call this semi-synthetic network \textbf{\SG}.


For examining intervention strategies in more detail, we use two location datasets developed for studying pandemics: \textbf{{\BBC}} and \textbf{{\Italy}}. The first records meetings between users of the BBC Pandemic Haslemere app over time, including pairwise distances with 5 minute intervals~\cite{kissler2020sparking}. The second reports temporal aggregated mobility metrics for each day's movement of population between Italian provinces based on smartphone user locations before and during the COVID-19 outbreak over 90 days~\cite{pepe2020covid}. The {\Italy} dataset also provides transition probabilities between provinces, which we directly use as the propagation probabilities for our T-IC process. 

Finally, we experiment on two social network datasets: 
\textbf{{\wiki}} \cite{leskovec2010predicting} and \textbf{{\cithep}}~\cite{leskovec2005graphs}. 
{\wiki} consists of user discussions on Wikipedia, with edges between users representing votes. {\cithep} encodes citation connections between research papers. 

The statistics of the constructed networks are in Table~\ref{statistics_of_constructed_graph}.

\input{plots/constructed_graph_statistics}

\input{05-probability-assignment}
\subsubsection{Baselines:\\}

\noindent
As there is no work examining sentinel and susceptible nodes using IC models, we look for comparable alternatives to our RSM and ESM solution sets. We select baselines in three 
groupings. The first consists of IC model-based methods (\textbf{Greedy-IM}~\cite{Kempe:2003}, \textbf{DIA}~\cite{ohsaka2016dynamic}), to demonstrate the superiority of T-IC for analyzing spread on evolving networks.  
The second is the traditional virus propagation method for finding the critical $k$ nodes to immunize to prevent an epidemic (\textbf{T-Immu}~\cite{prakash2010virus}). The final grouping covers simple heuristic-based methods (\textbf{Max-Deg}, \textbf{Random}).

Greedy-IM obtains the top-$k$ influential nodes using a greedy hill-climbing algorithm over $|T|$ time windows. Since it is infeasible for larger datasets, we run it only on the smallest {\BBC} and {\Italy} datasets, and apply Greedy-IM for each time window separately to calculate the corresponding spread over $T$ and average the results to select the best node.
DIA (Dynamic Influence Analysis) is designed for evolving networks with an updating index structure that only shows the graph connection at the latest timestamp. 
We select top-$k$ influential nodes at each time window using DIA, and report average results over the $|T|$ time windows.
We only run DIA and T-Immu on location datasets, since these include temporal information.
The Max-Deg algorithm selects the top-$k$ nodes in decreasing degree order. The Random algorithm selects $k$ nodes uniformly at random in a given graph, with average results presented after $20$ simulations.

\input{plots/table-split-1}

\input{plots/table-split-2}


We compare our solution sets with the influential sets from the alternatives with respect to the following performance measures:
(i)~Reverse spread from the solution set (ii)~Average number of activated nodes (expected spread) in the solution set (iii)~Binary success rate of detecting spread. The reverse spread $\phi(\cdot)$ is computed as defined in Section~\ref{section:4}. 
The binary success rate is the average number of times that there is at least one active node in the solution set during random T-IC processes. Reverse spread is expected to be correlated with binary success, as both relate to the effectiveness of the solution set (sentinel nodes) in covering/detecting spread in the network.
The expected spread, computed as the average number of activated nodes in the solution set, represents its susceptibility.
Specifically, we simulate $1000$ random T-IC processes to activate nodes in the network within time window $T$.
\subsection{Evaluation of results}
\label{evaluation_results}
\subsubsection{Performance with different solution set sizes $k$:\\} 
\noindent
Tables~\ref{normalized_performance_NYC-1} and~\ref{normalized_performance_NYC-2} summarize the comparative results on each dataset for solution sets ($S$) of sizes $k\!=\!10,20,\ldots,50$ with a time window of length $|T|\!=\!25$ days. All results are normalized for ease of comparison, i.e., the range of values between the minimum and maximum is mapped to $[0,10]$ to produce normalized value $x_{n}\!=\!\frac{x-x_{min}}{x_{max}-x_{min}}$, where $x$ is the original value, and $x_{min}$ and $x_{max}$ are the minimum and maximum values across all the methods on the same measure.   
For example, consider the normalized reverse spread on the {\NYC} dataset shown in Table~\ref{normalized_performance_NYC-1} . The value $0$ for DIA in this section at $k\!=\!10$ means that the reverse spread of DIA with $k\!=\!10$ is minimum among all methods from $k\!=\!10$ to $k\!=\!50$, while the value of $10$ for RSM at $k\!=\!50$ denotes that its reverse spread in this configuration is maximum across among all methods and solution set sizes.

The set returned by RSM collectively achieves the highest reverse spread coverage in all cases, which increases with increasing $k$ (solution set size). Without prior information about the initial seeds from where activation begins to spread, distributing limited resources (e.g., scarce/expensive medical tests) to these sentinel nodes (i.e., the nodes selected in $S$) increases the probability of detecting the spread at an early stage. 

By contrast, ESM selects all nodes having the highest probabilities of being activated during a random T-IC process, 
and thus best captures the largest expected spread out of all methods. ESM outperforms Max-Deg, which is often enforced in reality, by up to $82\%$ on {\NYC}. 
ESM is thus an effective targeted strategy for identifying the most susceptible nodes (e.g., for contact tracing or treatment). 


The binary success rate using RSM is the best 
for all datasets and $k$. 
Comparisons with T-Immu and DIA show that considering temporal properties while also preserving the overall graph structure is vital to select the ideal solution sets.
RSM consistently outperforms T-Immu (nearly $2$x better on {\SG}) 
despite having related objectives,
since T-Immu 
cannot capture time-varying transmission probabilities. 
RSM also drastically outperforms DIA, which is worse than Random, 
because DIA selects nodes of the evolving network based on an updating index which only remembers the latest probability assignment and fails to capture the globally optimal solution set over $T$.
The improvements (at least $10\%$ higher success rates in the worst case) over Max-Deg confirm that the dynamic topology of the network (captured by the RSM and ESM solutions with T-IC process) plays a much more significant role compared to the local connectivity (node degrees) when modeling the spread.

\input{plots/time-windows}


\subsubsection{Performance with different time window lengths $|T|$:\\}
\label{perofrmance_on_different_T_threedatasets}
\noindent
We evaluate the effect of varying the time window $T$ while keeping a constant solution set size $k\!=\!50$ on {\NYC}, {\Tokyo}, and {\SG} over one-day intervals up to $|T|=25$. We do the same on the datasets with pandemic background, i.e., {\BBC} and {\Italy}, where the propagation probability for {\Italy} is the real transmission probability of people moving between any two provinces over $|T|=90$ days, while for {\BBC} it captures infections over $|T|=3$ days.
Figures~\ref{all_datasets_performance} and \ref{all_datasets_ESMperformance} show that reverse spread, expected spread, and binary success all increase with $|T|$, as it allows more activations to take place. 
As expected, RSM has the best performance with respect to reverse spread and binary success rate in Figure~\ref{all_datasets_performance}, and ESM outperforms other methods in terms of having the best expected spread in Figure~\ref{all_datasets_ESMperformance}.

Only considering the node degrees is ineffective, particularly as propagation becomes more complex, e.g., on a large network and elongated time windows. 
DIA is jeopardized especially with smaller time windows as the overall optimality of the solution set is not guaranteed by the most recent snapshot of the graph. 
{\BBC} and {\Italy} datasets in Figures \ref{all_datasets_performance}--\ref{all_datasets_ESMperformance} also highlight how 
Greedy-IM cannot effectively capture the optimal global solution 
over multiple time windows. 

\subsubsection{Performance with different hyper-parameters:\\}
\label{hyperparams}
\noindent
The propagation probability in Equation \ref{definition:initial-probability} is proportional to hyper-parameters $a$, $b$, and $l$, while it is inversely proportional to $\rho_1$ and $\rho_2$. Furthermore, the propagation rates are sensitive to small changes in $\rho_1$ and $\rho_2$ since they directly influence the threshold of importance of proximity and population density respectively.  
The default values of hyper-parameters for all datasets are summarized in Table~\ref{hyper_parameter_datasets}.
We evaluate different hyper-parameter settings on {\SG}, {\Tokyo}, {\NYC}, and {\BBC} datasets. These hyper-parameters are not necessary for {\Italy} and the social network datasets ({\wiki} and {\cithep}), since the computation of $p$ does not involve them.

\input{plots/hyperparameter_table}

\input{plots/hyperparams}

We experiment with $\rho_1, \rho_2 \in [0.1,0.5]$, $ a, b \in [0.01,0.1]$, and $l \in [5,20]$.
We fix one of the parameters as the default value, then experiment with different values of the others, e.g., setting $b=0.05$ for {\SG}, and changing $\rho_2$ from $0.1$ to $0.5$.
Specifically, we modify $b$ and $\rho_2$ for {\SG}, {\Tokyo}, {\NYC}, where the population density at POIs is a relevant factor for the risk of propagation of disease. Meanwhile $a$, $\rho_1$, and $l$ are relevant for determining propagation risk in {\BBC}, and we modify them for this dataset in order to reflect different distance thresholds between infected and potential susceptible individuals.

We present the spread resulting from RSM and ESM solutions with different $a$, $\rho_1$, and $l$ on {\BBC} in Figure~\ref{BBC_hyperparameter_performance}.
The reverse spread and expected spread increase overall with increase in $a$ and $l$ while decrease with increase in $\rho_1$. This is intuitive, since the likelihood of propagation increases with larger values of $a$, and the longer distance threshold $l$ implies a high probability that the infection will successfully spread when individuals are in proximity to each other even though they are separated by some distance. Meanwhile, a large choice of $\rho_1$ will decrease the factor of the infection force that is brought about by the proximate contact.

Figure~\ref{hyperparameter_performance} shows the spread resulting from RSM and ESM with different $b$ and $\rho_2$ on the contact-based datasets {\SG}, {\Tokyo}, and {\NYC}. The performance fluctuates with the increase of $\rho_2$, while the performance gets large strictly with the increase of $b$. As shown in Equation~\ref{definition:probability-assignment}, $\rho_2$ is the factor that dictates the importance of the number of people $m$ in determining transmission risk. Hence, in the real-life application, it is important to select an appropriate value of $\rho_2$ to reflect how the density of population at a POI contributes to the spread.
It is important to note that we do not select our hyper-parameter values in a way that maximizes the spread and provides any undue advantage for our method. Rather, we choose values that are reasonable reflections of real-world scenarios.

\subsubsection{Running time efficiency:\\}
\noindent
While the solution quality improves with higher number of hypergraph nets generated (up to a certain point), there is an efficiency trade-off. 
We measure the running time of generating hypergraph nets by varying the desired number of nets $|N|$ and the number of time windows $|T|$ under consideration, as shown in Table~\ref{running_efficiency_against_different_hypernets} and
Table~\ref{running_efficiency_against_different_time_windows},  respectively.
Specifically, with $|N|$ increasing from $20K$ to $100K$ (for a fixed $|T|\!=\!5$) in Table~\ref{running_efficiency_against_different_hypernets}, there is a slight increase in running time (from $0.43$ to $2.18$ seconds for {\Tokyo}) demonstrating the efficient and scalable nature of our reachable set sampling based algorithm. The hypergraph construction time also increases with $|T|$ (for a fixed $|N|\!=\!20K$) in Table~\ref{running_efficiency_against_different_time_windows} due to a prolonged propagation process, which is especially evident in dense networks (e.g., {\Tokyo}).
Despite the increase in computation time and memory requirements when increasing $|N|$, we find that stable solution sets of sufficiently high quality are produced without the need for more than $20K$ hypergraphs.

T-IC can model large-scale individual-level contacts efficiently, whereas other solutions~\cite{prakash2010virus, holme2017three, ohsaka2016dynamic} are only feasible on small graphs.
For example, T-Immu is time-consuming due to repeated computations of the eigenvalue of the dynamic contact network structure which makes it not feasible for large-scale dataset (e.g., {\cithep}), and DIA can only consider the latest snapshot but not the global structure efficiently.

\input{plots/running_efficiency_hyperNets.tex}
\input{plots/running_efficiency_different_timeWindows.tex}
\input{plots/sampling_based_method_runningtime.tex}

We compare RSM with the commonly used IC model based method, Greedy-IM, in terms of the running time for selecting different size of solution set $S$ from $k\!=\!10$ to $k\!=\!50$ in Table~\ref{running_efficiency_samplingbased_on_BBC}. The running time of Greedy-IM grows quickly and gets infeasible with large dataset and long time windows, whereas our RSM and ESM running times grow much more slowly. Hence, for running time comparisons, we only experiment on the small dataset {\BBC}. We observe the running time of Greedy-IM increases sharply from 879s to 18948s with increase in $k$ at $T\!=\!3$, which makes it not applicable for large evolving networks.

\subsection{Intervention Strategies}
Applying our solution towards disease monitoring, we analyze the effect of several intervention strategies for reducing the spread of infectious diseases. 
We study targeted lockdown strategies and occupancy restrictions, identify superspreader venues, and examine the need for backward tracing.


\subsubsection{Intervention Spread Analysis:\\} 
\noindent
To simulate (partial) lockdown strategies, we reduce the edge connections in our network construction and analyze how the spread changes as a result of these dropped edges. We randomly select a seed set of $10$ nodes from which to simulate the T-IC process.
We use two intervention strategies to select the edge connections to drop (we drop $30\%$ of the total edges). The first is to randomly drop edges. The second is based on the priority of the venues, i.e., delete connections for the venues visited by more people. That is, the number of edges dropped is proportional to the number of connections to the venue. 
We perform $20$ simulations to get the average decrease in spread resulting from each of the two strategies. 
For {\NYC}, random deletion reduces $78\%$ spread while venue prioritization (on the top-$50$ busiest venues) achieves $83\%$ spread reduction.
Similarly, for {\SG}, random deletion reduces $26\%$ spread while an additional $4\%$ spread reduction is achieved by venue prioritizing the top-$50$ busiest venues.
For the less granular {\Italy} dataset, which does not have venue information, we prioritize the deletion of the top-$50$ densely connected provinces. Spread reduction is $49\%$ when using prioritization, while random deletion reduces $38\%$ spread.
Therefore, a targeted approach to lockdowns at specific POIs shows superior performance over random occupancy restrictions across all POIs.




\subsubsection{Backward contact tracing:\\}
\noindent
We calculate the contribution of backward traced nodes to the activations in the selected ESM solution set, in Table \ref{contact_tracing_contribution}. First, we select different sizes of solution set $S$ from $k\!=\!10$ to $k\!=\!50$. Considering the reverse reachable set of nodes from a given solution set for backward tracing, we identify the top spread contributors as the nodes that participate most frequently in activations. We find that this backward traced set of superspreader nodes account for $67.9\%$ to $95.0\%$ of the activations in $S$ on the {\BBC} dataset. For {\Tokyo}, they contribute $77.8\%$ to $96.1\%$. This skewed over-dispersion further points to the importance of backward contact tracing and need for suppressing superspreader events.

\input{plots/backward_tracing_table}

\subsubsection{Venues Analysis:\\}
\noindent
The {\Tokyo} and {\NYC} datasets also include venue categories which provide further insights for designing effective intervention strategies. For $|T|\!=\!25$, we observe that only $26$ to $83$ venues in NYC are visited by persons in solution sets selected by RSM when increasing $k$ from $10$ to $50$, while ESM, Max-Deg, and Random cover up to $3$x as many venues. 
For {\Tokyo} and {\NYC}, an analysis of the categories of venues visited reveals transportation hubs (including airport, subway, and train station), restaurants, bars, and coffee shops as  superspreaders in the solutions sets, with transportation hubs in particular having an out-sized impact when increasing the set size $k$ of infected individuals. 




%% file: plots/constructed_graph_statistics.tex
\begin{table}[htp]
\caption{Dataset properties}
\label{statistics_of_constructed_graph}
\begin{tabular}{l|r|r|r}
\hline
\centering
Dataset   & \multicolumn{1}{l|}{\#Nodes} & \multicolumn{1}{l|}{\#Edges} & \multicolumn{1}{l}{Max degree}\\ \hline
\NYC       & 876                           & 18270                         & 147             \\ 
\Tokyo     & 765                           & 102018                          & 311              \\ 
\SG   & 2000                           & 57530                        & 56              \\ 
\BBC   &  469                           & 205662                          & 1506                 \\ 
\Italy   &  111                           & 235190                          & 6808                     \\ 
\wiki       & 8297                           & 103689                         & 1167
\\ 
\cithep & 34546                         & 841798                        & 846   
\\ 
\hline
\end{tabular}
\end{table}

%% file: 05-probability-assignment.tex
\subsubsection{Propagation probability setting:\\}
\label{probability_assignment}
\noindent
Each edge is assigned its corresponding probability of propagation (e.g., transmission of disease) from one node to the other, based on the needs of the specific application dataset.

\noindent$\bullet$ For social network 
datasets, the propagation probability is assigned following the common practice in influence modeling studies~\cite{chen2010scalable} of using a uniform distribution. We randomly assign the edges of the network to a discrete time interval in $[0, T]$, and sample $p\in[0,0.3]$ for each of the edges.

\noindent$\bullet$ For location-based contact networks and \BBC, we utilize a domain-informed probability assignment. The {\Italy}  dataset directly uses the provided transmission probability between connected nodes.  
Recent epidemiological studies quantify how transmission rates are related with the distance between the individuals as well as the overall population density at the location~\cite{kissler2020sparking,mittal2020mathematical,gross2020spatio,anastasiou2021astro}. Based on these, we calculate the propagation probability $p$ of a connection from node $u$ to node $v$ at time interval $t$  using Equation~ \ref{definition:initial-probability}, to incorporate knowledge of virus spreading characteristics:

\begin{equation}
p_{t}(u,v) = 1 - exp(-\sum_{t'}\lambda_{u,v}(t'))
\label{definition:initial-probability}
\end{equation}
where $\lambda_{u,v}$ is a factor denoting the ``force of infection'' (the larger the value of $\lambda_{u,v}$, the greater is the transmission probability between $u$ and $v$)~\cite{kissler2020sparking}, and
$t' \in (t-t_0,t]$ indicates the relevant duration of time up to the current time interval $t$. The latter is governed by $t_{0}$, the duration for which historic infection force is considered, since the transmission probability is decided by the accumulated infection force over $t'$.


Therefore, a minimal expression for $\lambda_{u,v}(t)$ must consider the distance from $u$ to $v$ and the population density at the venue to determine risk, and is formulated based on the literature as:
\begin{equation}
\lambda_{u,v}(t) =
 ae^{-{d_{u,v,t}}{\rho_{1}}} + be^{-{{m}^{-1}\rho_{2}}}
\label{definition:probability-assignment}
\end{equation}
where $d_{u,v,t}$ is the distance between $u$ and $v$ at time interval $t$ (based on their location data), 
$m$ is the number of people located at the same venue, and $a$, $b$, $\rho_1$ and $\rho_2$ are hyper-parameters. In line with~\cite{kissler2020sparking}, to realistically simulate spread, we use default values of $\rho_{1}=\rho_{2}=0.1$, $a=0.05$, and choose $b=0.05$ (for {\NYC} and {\SG} datasets) or $b=0.01$ (for {\Tokyo} dataset due to its dense connectivity). When $d_{i,j} \textgreater l$ (distance threshold) or when the dataset has no such proximity information, the contribution to infection force is assumed to be zero (i.e., $a=0$).

\input{plots/plot_probability_assignment}

Figure~\ref{probability_assignment_plot} shows an example of our dynamic probability assignment for the {\BBC} network, demonstrating the accumulation of infection force and the changing propagation probability as distances vary with time during the interactions between three node-pairs (i.e., $(u_1,v_1)$, $(u_2,v_2)$, and $(u_3,v_3)$ in Figures~\ref{bbc_T_distance}--\ref{BBC_prob_plot}). Here, we choose a distance threshold of $l=5$ meters as shown in Figure~\ref{distance_threshold}, so there is no infection force once the distance between a node-pair is greater than 5 meters. Figure~\ref{bbc_T_distance} shows the varying distance between contact node-pairs every $5$ minutes over three days. The {\BBC} data covers $16$ hours of each day, and for simplicity of illustration we ignore the remaining hours of each day on the x-axis of Figures~\ref{bbc_T_distance}--\ref{BBC_prob_plot}. We select $t_0=1$ day, i.e., the transmission probability at time $t$ is decided by the past $1$ day's interactions. Figure~\ref{bbc_T_accumulated_force} is the corresponding accumulated infection force of Figure~\ref{bbc_T_distance}, which is computed as $\sum_{t'}\lambda_{u,v}(t')$ using Equation~\ref{definition:probability-assignment} to calculate $\lambda$. 
The trend of the propagation probability in Figure~\ref{BBC_prob_plot} is the same as in Figure~\ref{bbc_T_accumulated_force} because the propagation probability is proportional to the accumulated infection force, as shown in Equation~\ref{definition:initial-probability}. 
A more detailed exploration of the influence of the various hyper-parameter settings for Equation~\ref{definition:probability-assignment} can be found in Section~\ref{hyperparams}.
While we experiment with various hyper-parameter settings for disease modeling applications, these may be easily customized to incorporate the domain findings on transmission risks, e.g.,~\cite{bazant2020guideline,nordsiek2021risk} (based on contact duration, venue size and occupancy rates, activity type, ventilation, and other factors) as an orthogonal scope of work.

%% file: plots/plot_probability_assignment.tex


\begin{figure}[tb]
\centering  %
\subfloat[Distance-based $\lambda_{u,v}(t)$]
{\label{distance_threshold}\includegraphics[width=0.45\textwidth]{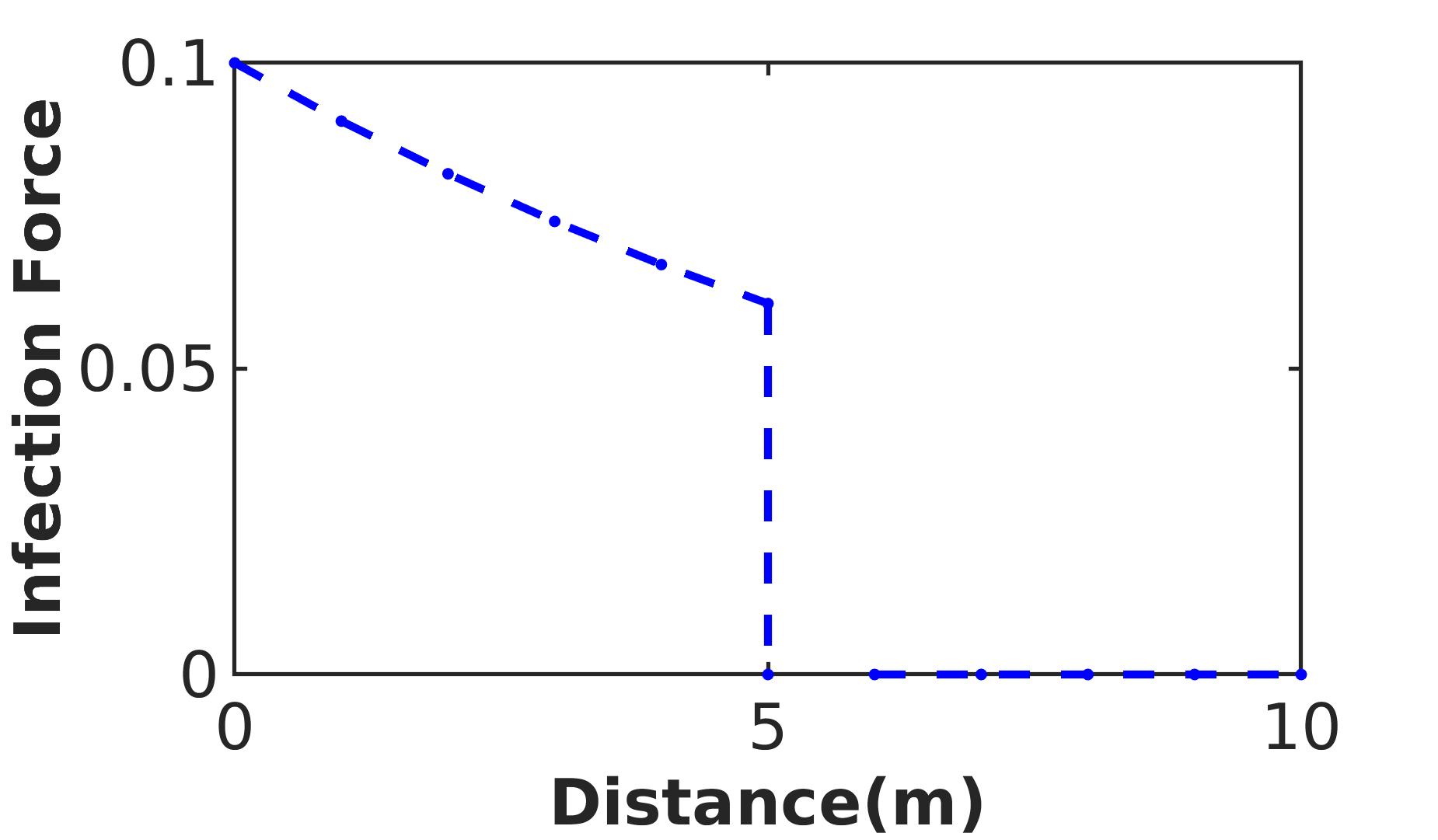}}
\subfloat[Distance at different $t$]
{\label{bbc_T_distance}\includegraphics[width=0.45\textwidth]{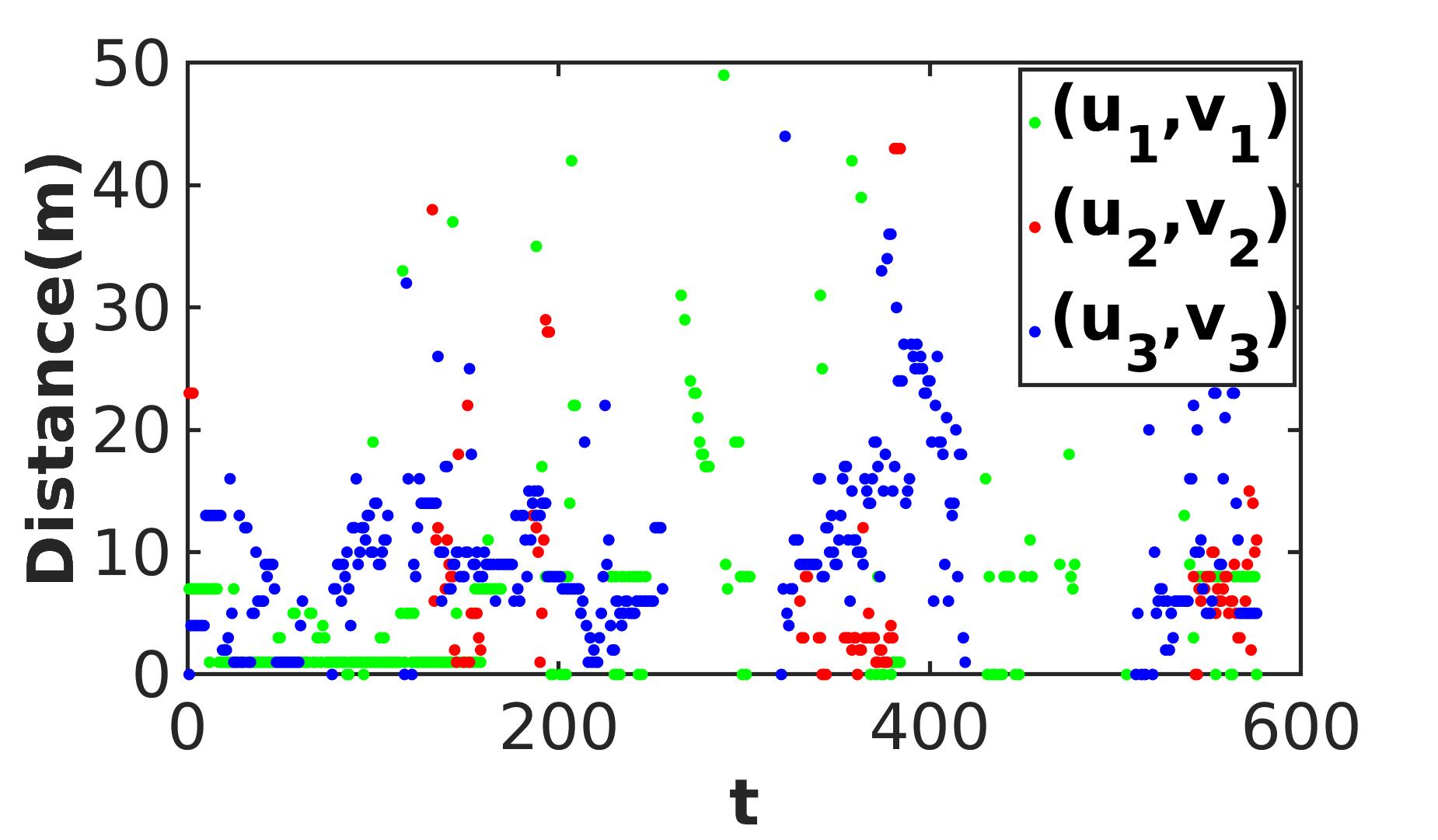}}
\\
\subfloat[Accumulated infection force with $t$]
{\label{bbc_T_accumulated_force}\includegraphics[width=0.45\textwidth]{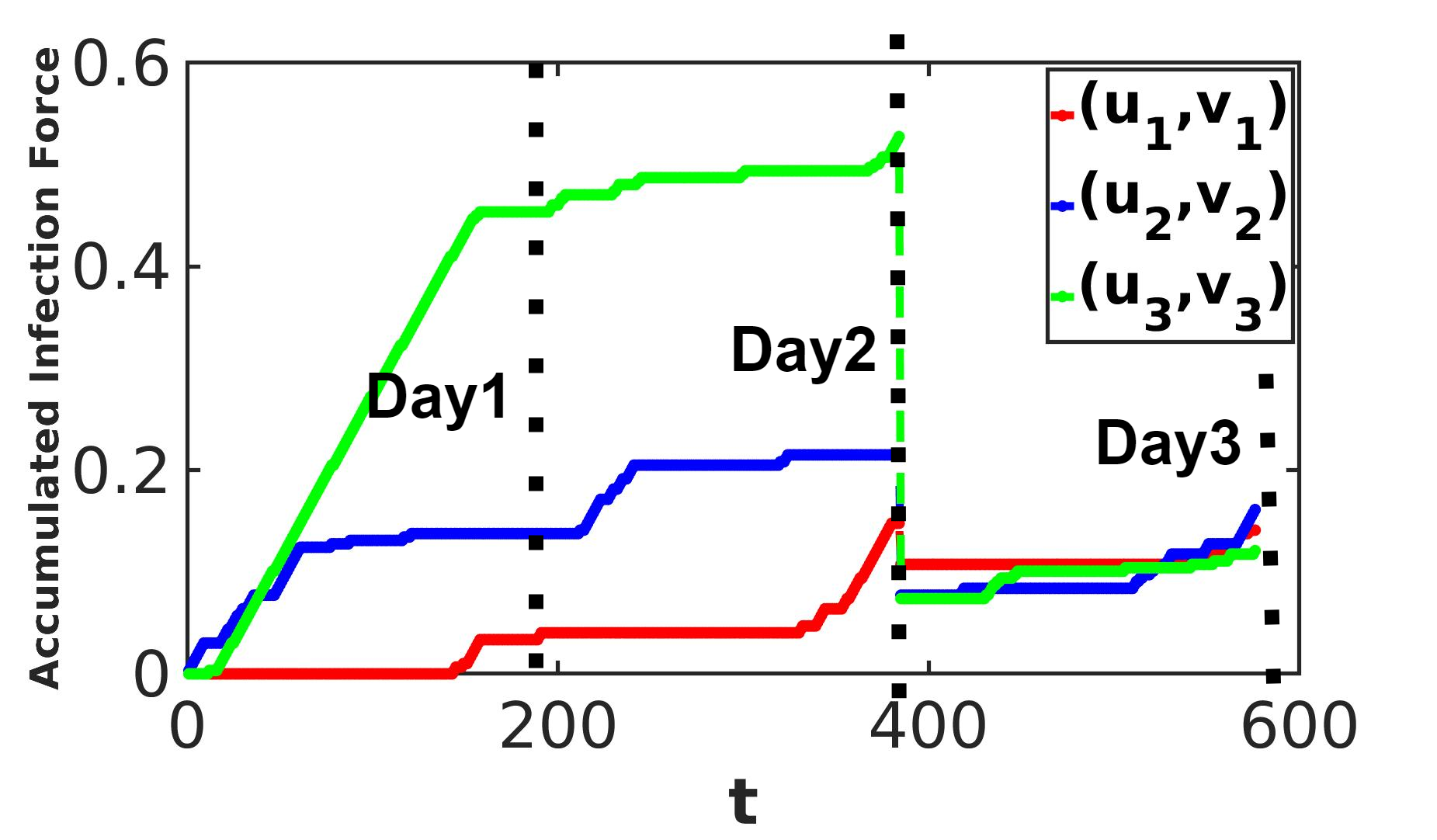}}
\subfloat[Propagation probability with $t$]
{\label{BBC_prob_plot}\includegraphics[width=0.45\textwidth]{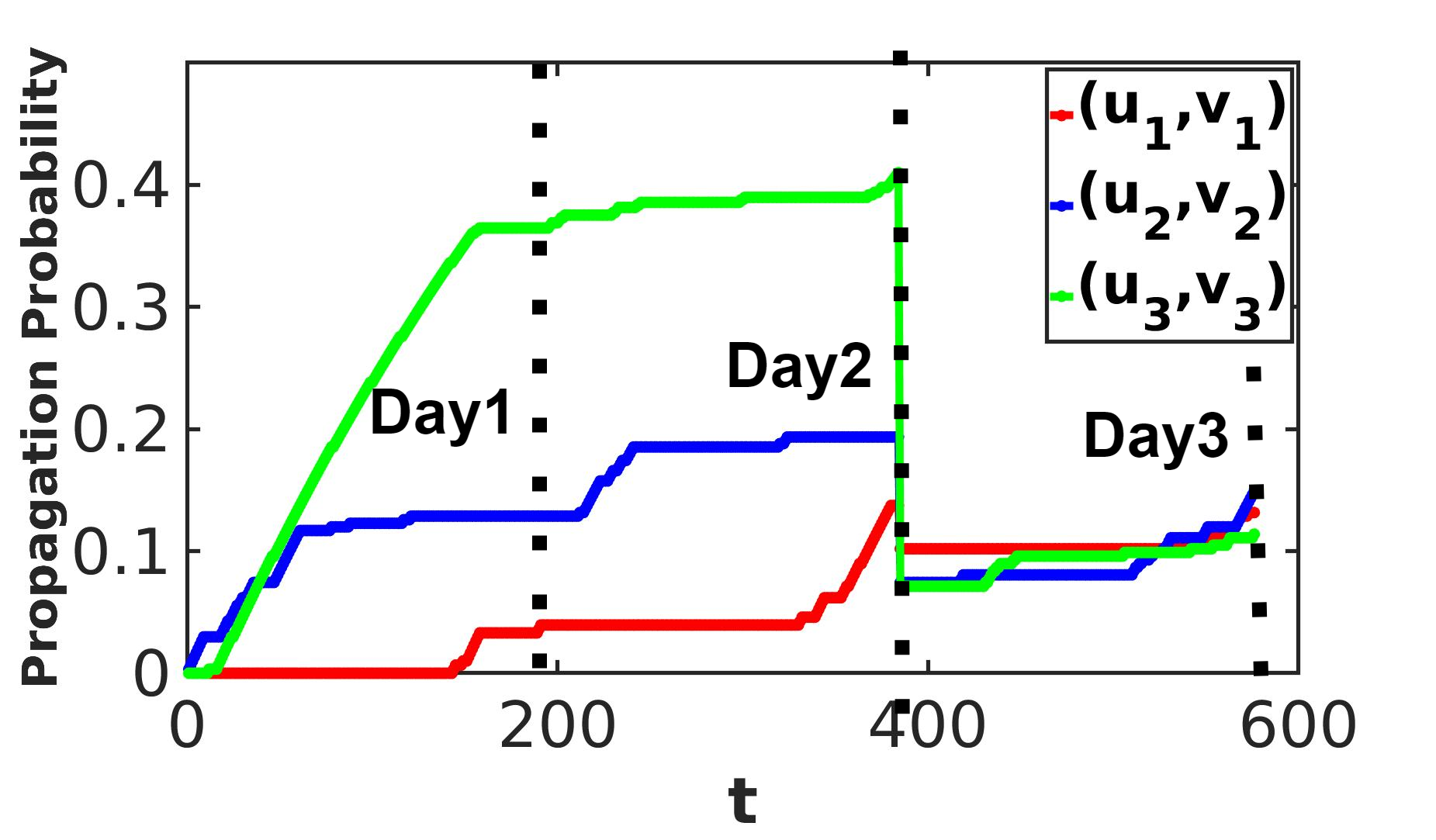}}
\caption{Propagation probability example for Haslemere}
\label{probability_assignment_plot}
\end{figure}


%% file: plots/table-split-1.tex
\begin{table}[t]
\caption{Normalized performance (reverse spread and binary success rate) at $|T|\!=\!25$ with different sizes of solution set $|S|\!=\!k$}
\label{normalized_performance_NYC-1}
\begin{tabular}
{ll|lllll|lllll}
\hline
Dataset & Method & \multicolumn{5}{c|}{Reverse Spread}  & \multicolumn{5}{c}{Binary Success Rate} \\ 
\multicolumn{2}{c|}{k}       & 10                                    & 20              & 30              & 40              & 50        & 10                   & 20                   & 30                   & 40                   & 50  \\ 
\hline

  
{\multirow{5}{*}{\NYC}} & RSM                   & \textbf{7.9}  & \textbf{8.5}                       & \textbf{9.0}   & \textbf{9.5}   & \textbf{10}        & \textbf{8.6}           & \textbf{7.6}            & \textbf{9.1}            & \textbf{8.4}           & \textbf{10}                \\
& T-Immu               & 5.0       & \multicolumn{1}{l}{8.2}   & 8.3                       & 8.7   & 9.2     & 7.4                & 7.1            & 7.5            & 6.9             & 8.6\\
& DIA               & 0.0       & \multicolumn{1}{l}{1.0}   & 2.4                                 & 3.2   & 4.1     & 0.0                & 1.1            & 3.1            & 4.2             & 4.8 \\ 
& Max-Deg       & 7.5  & 8.0                       & 8.4                       & 8.8   & 9.2      & 7.5                & 6.7          & 7.4            & 6.8            & 8.3            \\
& Random               & 4.0       & \multicolumn{1}{l}{6.3}   & 7.9                       & 8.7   & 9.2       & 2.1                & 6.0            & 6.6            & 5.7             & 6.6 \\ \hline
 
{\multirow{5}{*}{\Tokyo}} & RSM    & \textbf{8.4}   & \textbf{8.9}                      & \textbf{9.3}                       & \textbf{9.7}   & \textbf{10}                     & \textbf{8.4}           & \textbf{9.1}            & \textbf{9.0}            & \textbf{9.2}           & \textbf{10}                \\
& T-Immu               & 8.3       & \multicolumn{1}{l}{8.8}   & 9.1                       & 9.3   & 9.7      & 7.9                & 8.4            & 8.2            & 8.5             & 8.9
\\ 
& DIA               & 0.0       & \multicolumn{1}{l}{0.9}   & 2.0                                 & 3.0   & 4.0       & 0.0                & 0.7            & 1.4            & 2.7             & 4.0 \\ 
& Max-Deg       & 7.9 & 8.5                       & 9.0                       & 9.3   & 9.7      & 7.6                & 8.1           & 7.8            & 8.4            & 9.0            \\
& Random               & 6.6       & \multicolumn{1}{l}{8.1}   & 8.7                       & 9.3   & 9.7      & 5.5                & 6.9            & 6.5            & 7.3             & 7.8 
\\
\hline
{\multirow{5}{*}{\SG}} & RSM                   & \textbf{3.1}  & \textbf{5.4}                       & \textbf{7.3}   & \textbf{8.6}   & \textbf{10}       &  \textbf{3.1}           & \textbf{5.3}            & \textbf{7.0}            &  \textbf{8.2}           & \textbf{10}                \\

& T-Immu               & 1.7       & \multicolumn{1}{l}{3.2}   & 4.8                       & 6.3   & 7.3     & 1.7                & 2.7            & 3.8            & 4.0             & 5.6
\\ 
& DIA               & 0.0       & \multicolumn{1}{l}{0.1}   & 0.1                                 & 0.2   & 0.2      & 0.1                & 0.1            & 0.2            & 0.2             & 0.2 \\ 
& Max-Deg       & 2.0   & 3.3                       & 5.1                       & 6.3   & 7.5        & 2.1                & 2.6           & 4.2            & 4.2            & 6.6            \\
& Random               & 1.7       & \multicolumn{1}{l}{3.4}   & 4.9                       & 6.5   & 7.8      & 0.0                & 0.1            & 0.6            & 0.7             & 1.4 
\\
\hline
 
{\multirow{3}{*}{\wiki}} & RSM                             & \textbf{6.4}    & \textbf{7.9}    & \textbf{8.8}                        & \textbf{9.5}                     & \textbf{10}          & \textbf{6.6}       & \textbf{8.5}      & \textbf{8.6}        & \textbf{8.8}              & \textbf{10}             \\
& Max-Deg & 0.0    & 0.8    & \multicolumn{1}{l}{1.9}    & 2.9                      &  4.4     & 0.6         & \multicolumn{1}{l}{1.8}        &  2.8      &  3.4       & 4.1       \\
& Random                             & 0.1        & 1.1   & 1.8                       & \multicolumn{1}{l}{2.5}   & 3.0      &  0.0         &  0.6        &  1.3     & 1.7        & 1.9
\\ \hline

{\multirow{3}{*}{\cithep}} & RSM                             & \textbf{3.5}    & \textbf{5.7}    &\textbf{7.5}                        & \textbf{8.8}                       & \textbf{10}         & \textbf{3.9}         & \textbf{4.7}                            & \textbf{7.0}        & \textbf{8.0}       & \textbf{10}            \\
& Max-Deg & 0.1   & 0.2    & \multicolumn{1}{l}{0.3}    & 0.6                      & 0.9     & 0.0        & \multicolumn{1}{l}{0.2}        & 0.4       & 0.6        & 1.0        \\
& Random                             & 0.0        & 0.0    & 0.0                        & \multicolumn{1}{l}{0.4}   & 0.5    & 0.0             & 0.0                            & 0.0        & 0.0        & 0.0    
\\ \hline

\end{tabular}
\end{table}

%% file: plots/table-split-2.tex
\begin{table}[t]
\caption{Normalized performance (expected spread) at $|T|\!=\!25$ with different sizes of solution set $|S|\!=\!k$}
\label{normalized_performance_NYC-2}
\begin{tabular}
{ll|lllll}
\hline
Dataset & Method &  \multicolumn{5}{c}{Expected Spread}  \\ 
\multicolumn{2}{c|}{k}      & 10                   & 20                   & 30                   & 40                   & 50  \\ 
\hline
  
{\multirow{5}{*}{\NYC}} & ESM             & \textbf{2.3}        & \textbf{3.8}       & \textbf{6.1}       & \textbf{7.1}       & \textbf{10}           \\
& T-Immu            & 0.4        & 0.7      & 1.7       & \multicolumn{1}{l}{2.0}       & 2.3    \\
& DIA                 & 0.0        & 0.0      & 0.1       & \multicolumn{1}{l}{0.2}       & 0.3       \\ 
& Max-Deg       & 0.4            & 0.5       & 0.9      & 1.1                           & 1.8              \\
& Random               & 0.1        & 0.4      & 0.7       & \multicolumn{1}{l}{0.9}       & 1.2       \\ \hline

{\multirow{5}{*}{\Tokyo}} & ESM            & \textbf{2.2}        & \textbf{4.5}       & \textbf{5.9}       & \textbf{7.8}       & \textbf{10}       \\
& T-Immu              & 1.1        & 2.0      & 3.5       & \multicolumn{1}{l}{5.0}       & 6.8     \\ 
& DIA                & 0.0        & 0.0      & 0.1       & \multicolumn{1}{l}{0.2}       & 0.3      \\ 
& Max-Deg      & 0.9            & 1.7       & 1.9      & 2.9                           & 3.6                \\
& Random              & 0.3        & 0.9      & 1.2       & \multicolumn{1}{l}{1.7}       & 2.3      \\
\hline
{\multirow{5}{*}{\SG}} & ESM               & \textbf{2.1}        & \textbf{3.9}       & \textbf{5.9}       & \textbf{6.9}       & \textbf{10}                 \\
& T-Immu             & 0.9        & 1.4      & 2.2       & \multicolumn{1}{l}{2.3}       & 3.4     \\ 
& DIA           & 0.0        & 0.1      & 0.1       & \multicolumn{1}{l}{0.1}       & 0.1      \\ 
& Max-Deg    & 1.0            & 1.3       & 2.3      & 2.5                           & 3.8                 \\
& Random                 & 0.0        & 0.0      & 0.3       & \multicolumn{1}{l}{0.4}       & 0.7      \\
\hline
 
{\multirow{3}{*}{\wiki}} & ESM           & \textbf{2.8}     & \textbf{6.7}      & \textbf{8.4}      & \textbf{9.6}     & \textbf{10}            \\
& Max-Deg   & 0.1      & 0.3      & 0.5                          & 0.8                       & 1.4         \\
& Random                          & 0.0          & 0.1      & 0.2                         & 0.3                         & 0.3   \\ 
\hline
{\multirow{3}{*}{\cithep}} & ESM      & \textbf{3.5}      & \textbf{5.1}     & \textbf{6.7}      & \textbf{8.4}     & \textbf{10}          \\
& Max-Deg  & 0.0      & 0.2      & 0.3                         & 0.5                         & 0.7            \\
& Random       & 0.0          & 0.0      & 0.0                         & 0.0                         & 0.0    \\ 
\hline
\end{tabular}
\end{table}

%% file: plots/time-windows.tex
\begin{figure}[h!]
\centering
\scalebox{0.8}{
\fbox{
$\qed$ \LARGE\emph{RSM}
\ \ \ \tikz \draw[black,fill=black] (0.15,0.15); $\blacksquare$   \LARGE \emph{T-Immu}
\ \ \ \tikz \draw[black,fill=black] (0.5,0);  $\oplus$ \LARGE$\emph{DIA}$
 \ \ \ \tikz \draw[black,fill=black] (0.5,0);  $\star$  \LARGE$\emph{Greedy-IM}$
   \ \ \ \tikz \draw[black,fill=black] (0.5,0) circle (.5ex);  \LARGE$\emph{Max-Deg}$ 
 \ \ \ \tikz \draw[black] (0,0) circle (.5ex); \LARGE$\emph{Random}$
 }
}

{\input{plots/different_timewindow_overall_RSM} \\ {\vspace{-2mm}(a) Reverse Spread}}
\quad
{\input{plots/different_timewindow_overall_binary_success_rate} \\ {\vspace{-2mm}(b) Binary Success Rate}}
\\ 
\vspace{-2mm}
    \caption{\textit{Normalized performance (reverse spread, binary success rate) at $k\!=\!50$ with different lengths of time window $|T|$}}
    \label{all_datasets_performance}   
\end{figure}
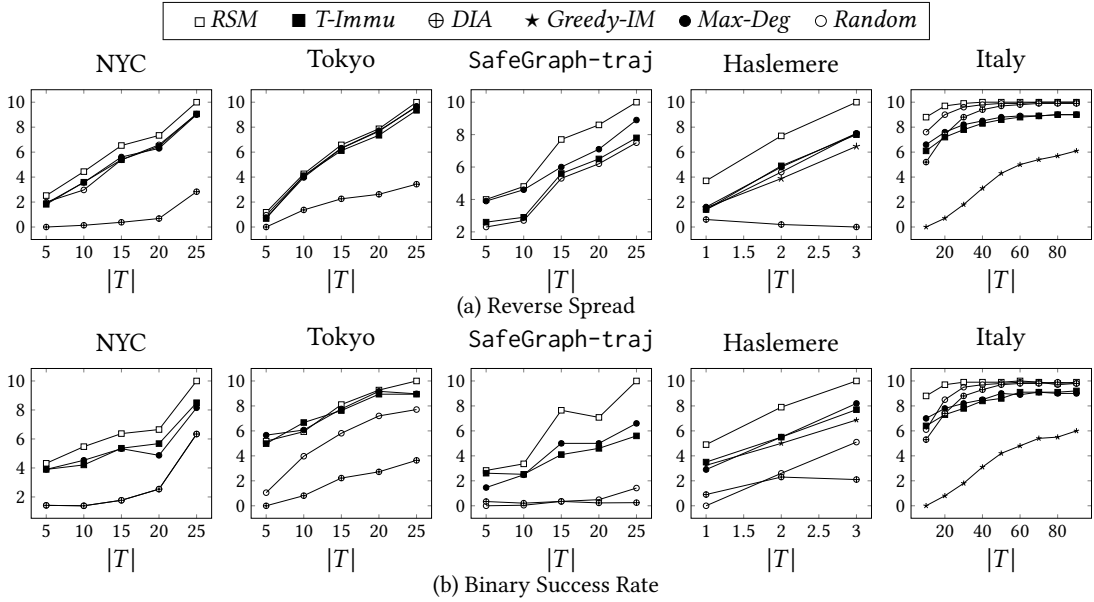

\begin{figure}[h!]
\centering
\scalebox{0.8}{
\fbox{
 $\triangle$ \LARGE\emph{ESM}
 \ \ \ \tikz \draw[black,fill=black] (0.5,0);  $\oplus$ \LARGE$\emph{DIA}$
 \ \ \ \tikz \draw[black,fill=black] (0.5,0);  $\star$  \LARGE$\emph{Greedy-IM}$
   \ \ \ \tikz \draw[black,fill=black] (0.5,0) circle (.5ex);  \LARGE$\emph{Max-Deg}$ 
 \ \ \ \tikz \draw[black] (0,0) circle (.5ex); \LARGE$\emph{Random}$
 }
}
\\
{\label{NYC_different_T_figure}
   	\input{plots/different_timewindow_overall_ESM}} \\
    \caption{\textit{Normalized expected spread at $k\!=\!50$ with different lengths of time window $|T|$}}
    \label{all_datasets_ESMperformance}   
\end{figure}
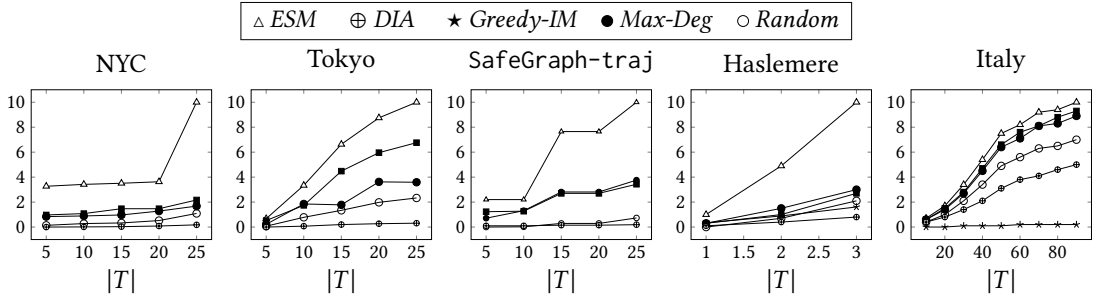

%% file: plots/different_timewindow_overall_RSM.tex
	\begin{tikzpicture}[thick,scale=0.35, every node/.style={scale=2.0}]
\begin{groupplot}[group style={group size=5 by 1, vertical sep= 2.5cm, horizontal sep= 1.5cm}]

\nextgroupplot [	xlabel={\huge $|T|$},	
title= {\huge NYC }]

\addplot[solid, every mark/.append style={solid, fill=white, scale=1.5}, mark=square*]  table[x=i,y=Gunduz] {plots/new_NYC_reverse_influence};


\addplot[solid, every mark/.append style={solid, fill=black, scale=1.5}, mark=otimes*]  table[x=i,y=Max-connection] {plots/new_NYC_reverse_influence};

\addplot[solid, every mark/.append style={solid, fill=white, scale=1.5}, mark=o]  table[x=i,y=Random] {plots/new_NYC_reverse_influence};	

\addplot[solid, every mark/.append style={solid, fill=black, scale=1.5}, mark=square*]  table[x=i,y=T-Immu] {plots/new_NYC_reverse_influence};

\addplot[solid, every mark/.append style={solid, fill=white, scale=1.5}, mark=oplus*]  table[x=i,y=DIA] {plots/new_NYC_reverse_influence};

\nextgroupplot [	xlabel={\huge $|T|$},	
title= {\huge Tokyo }]

\addplot[solid, every mark/.append style={solid, fill=white, scale=1.5}, mark=square*]  table[x=i,y=Gunduz] {plots/TKY_reverse_influence};


\addplot[solid, every mark/.append style={solid, fill=black, scale=1.5}, mark=otimes*]  table[x=i,y=Max-connection] {plots/TKY_reverse_influence};

\addplot[solid, every mark/.append style={solid, fill=white, scale=1.5}, mark=o]  table[x=i,y=Random] {plots/TKY_reverse_influence};	

\addplot[solid, every mark/.append style={solid, fill=white, scale=1.5}, mark=oplus*]  table[x=i,y=DIA] {plots/TKY_reverse_influence};

\addplot[solid, every mark/.append style={solid, fill=black, scale=1.5}, mark=square*]  table[x=i,y=TImmu] {plots/TKY_reverse_influence};

\nextgroupplot [	xlabel={\huge $|T|$},	
title= {\huge \SG }]

\addplot[solid, every mark/.append style={solid, fill=white, scale=1.5}, mark=square*]  table[x=i,y=Gunduz] {plots/new_safegraph_reverse_influence};

\addplot[solid, every mark/.append style={solid, fill=black, scale=1.5}, mark=otimes*]  table[x=i,y=Max-connection] {plots/new_safegraph_reverse_influence};

\addplot[solid, every mark/.append style={solid, fill=white, scale=1.5}, mark=o]  table[x=i,y=Random] {plots/new_safegraph_reverse_influence};

\addplot[solid, every mark/.append style={solid, fill=black, scale=1.5}, mark=square*]  table[x=i,y=TImmu] {plots/new_safegraph_reverse_influence};


\nextgroupplot [	xlabel={\huge $|T|$},	
title= {\huge Haslemere }]

\addplot[solid, every mark/.append style={solid, fill=white, scale=1.5}, mark=square*]  table[x=i,y=Gunduz] {plots/BBC_reverse_influence};


\addplot[solid, every mark/.append style={solid, fill=black, scale=1.5}, mark=otimes*]  table[x=i,y=Max-connection] {plots/BBC_reverse_influence};

\addplot[solid, every mark/.append style={solid, fill=white, scale=1.5}, mark=o]  table[x=i,y=Random] {plots/BBC_reverse_influence};	

\addplot[solid, every mark/.append style={solid, fill=white, scale=1.5}, mark=oplus*]  table[x=i,y=DIA] {plots/BBC_reverse_influence};

\addplot[solid, every mark/.append style={solid, fill=black, scale=1.5}, mark=square*]  table[x=i,y=TImmu] {plots/BBC_reverse_influence};

\addplot[solid, every mark/.append style={solid, fill=black, scale=2}, mark=star]  table[x=i,y=greedy-IM] {plots/BBC_reverse_influence};

\nextgroupplot [	xlabel={\huge $|T|$},	
title= {\huge Italy }]

\addplot[solid, every mark/.append style={solid, fill=white, scale=1.5}, mark=square*]  table[x=i,y=Gunduz] {plots/Italy_reverse_influence};


\addplot[solid, every mark/.append style={solid, fill=black, scale=1.5}, mark=otimes*]  table[x=i,y=Max-connection] {plots/Italy_reverse_influence};

\addplot[solid, every mark/.append style={solid, fill=white, scale=1.5}, mark=o]  table[x=i,y=Random] {plots/Italy_reverse_influence};	

\addplot[solid, every mark/.append style={solid, fill=white, scale=1.5}, mark=oplus*]  table[x=i,y=DIA] {plots/Italy_reverse_influence};

\addplot[solid, every mark/.append style={solid, fill=black, scale=1.5}, mark=square*]  table[x=i,y=TImmu] {plots/Italy_reverse_influence};

\addplot[solid, every mark/.append style={solid, fill=black, scale=1.5}, mark=star]  table[x=i,y=greedy-IM] {plots/Italy_reverse_influence};

\end{groupplot}
		
\end{tikzpicture}

%% file: plots/different_timewindow_overall_binary_success_rate.tex
	\begin{tikzpicture}[thick,scale=0.35, every node/.style={scale=2.0}]
\begin{groupplot}[group style={group size=5 by 1, vertical sep= 2.5cm, horizontal sep= 1.5cm}]

\nextgroupplot [	xlabel={\huge $|T|$},	
title= {\huge NYC }]

\addplot[solid, every mark/.append style={solid, fill=white, scale=1.5}, mark=square*]  table[x=i,y=Gunduz] {plots/new_NYC_normalized_binary_success_rate};


\addplot[solid, every mark/.append style={solid, fill=black, scale=1.5}, mark=otimes*]  table[x=i,y=Max-connection] {plots/new_NYC_normalized_binary_success_rate};

\addplot[solid, every mark/.append style={solid, fill=white, scale=1.5}, mark=o]  table[x=i,y=Random] {plots/new_NYC_normalized_binary_success_rate};	

\addplot[solid, every mark/.append style={solid, fill=black, scale=1.5}, mark=square*]  table[x=i,y=T-Immu] {plots/new_NYC_normalized_binary_success_rate};

\addplot[solid, every mark/.append style={solid, fill=white, scale=1.5}, mark=oplus*]  table[x=i,y=DIA] {plots/new_NYC_normalized_binary_success_rate};

\nextgroupplot [	xlabel={\huge $|T|$},	
title= {\huge Tokyo }]

\addplot[solid, every mark/.append style={solid, fill=white, scale=1.5}, mark=square*]  table[x=i,y=Gunduz] {plots/TKY_normalized_binary_success_rate};


\addplot[solid, every mark/.append style={solid, fill=black, scale=1.5}, mark=otimes*]  table[x=i,y=Max-connection] {plots/TKY_normalized_binary_success_rate};

\addplot[solid, every mark/.append style={solid, fill=white, scale=1.5}, mark=o]  table[x=i,y=Random] {plots/TKY_normalized_binary_success_rate};

\addplot[solid, every mark/.append style={solid, fill=black, scale=1.5}, mark=square*]  table[x=i,y=TImmu] {plots/TKY_normalized_binary_success_rate};

\addplot[solid, every mark/.append style={solid, fill=white, scale=1.5}, mark=oplus*]  table[x=i,y=DIA] {plots/TKY_normalized_binary_success_rate};

\nextgroupplot [	xlabel={\huge $|T|$},	
title= {\huge \SG }]

\addplot[solid, every mark/.append style={solid, fill=white, scale=1.5}, mark=square*]  table[x=i,y=Gunduz] {plots/new_safegraph_normalized_binary_success_rate};


\addplot[solid, every mark/.append style={solid, fill=black, scale=1.5}, mark=otimes*]  table[x=i,y=Max-connection] {plots/new_safegraph_normalized_binary_success_rate};

\addplot[solid, every mark/.append style={solid, fill=white, scale=1.5}, mark=o]  table[x=i,y=Random] {plots/new_safegraph_normalized_binary_success_rate};

\addplot[solid, every mark/.append style={solid, fill=black, scale=1.5}, mark=square*]  table[x=i,y=TImmu] {plots/new_safegraph_normalized_binary_success_rate};

\addplot[solid, every mark/.append style={solid, fill=white, scale=1.5}, mark=oplus*]  table[x=i,y=DIA] {plots/new_safegraph_normalized_binary_success_rate};

\nextgroupplot [	xlabel={\huge $|T|$},	
title= {\huge Haslemere }]

\addplot[solid, every mark/.append style={solid, fill=white, scale=1.5}, mark=square*]  table[x=i,y=Gunduz] {plots/BBC_normalized_binary_success_rate};


\addplot[solid, every mark/.append style={solid, fill=black, scale=1.5}, mark=otimes*]  table[x=i,y=Max-connection] {plots/BBC_normalized_binary_success_rate};

\addplot[solid, every mark/.append style={solid, fill=white, scale=1.5}, mark=o]  table[x=i,y=Random] {plots/BBC_normalized_binary_success_rate};

\addplot[solid, every mark/.append style={solid, fill=black, scale=1.5}, mark=square*]  table[x=i,y=TImmu] {plots/BBC_normalized_binary_success_rate};

\addplot[solid, every mark/.append style={solid, fill=white, scale=1.5}, mark=oplus*]  table[x=i,y=DIA] {plots/BBC_normalized_binary_success_rate};

\addplot[solid, every mark/.append style={solid, fill=black, scale=1.5}, mark=star]  table[x=i,y=greedy-IM] {plots/BBC_normalized_binary_success_rate};

\nextgroupplot [	xlabel={\huge $|T|$},	
title= {\huge Italy }]

\addplot[solid, every mark/.append style={solid, fill=white, scale=1.5}, mark=square*]  table[x=i,y=Gunduz] {plots/Italy_normalized_binary_success_rate};


\addplot[solid, every mark/.append style={solid, fill=black, scale=1.5}, mark=otimes*]  table[x=i,y=Max-connection] {plots/Italy_normalized_binary_success_rate};

\addplot[solid, every mark/.append style={solid, fill=white, scale=1.5}, mark=o]  table[x=i,y=Random] {plots/Italy_normalized_binary_success_rate};

\addplot[solid, every mark/.append style={solid, fill=black, scale=1.5}, mark=square*]  table[x=i,y=TImmu] {plots/Italy_normalized_binary_success_rate};

\addplot[solid, every mark/.append style={solid, fill=white, scale=1.5}, mark=oplus*]  table[x=i,y=DIA] {plots/Italy_normalized_binary_success_rate};

\addplot[solid, every mark/.append style={solid, fill=black, scale=1.5}, mark=star]  table[x=i,y=greedy-IM] {plots/Italy_normalized_binary_success_rate};
\end{groupplot}
		
\end{tikzpicture}

%% file: plots/different_timewindow_overall_ESM.tex
\begin{tikzpicture}[thick,scale=0.35, every node/.style={scale=2.0}]
\begin{groupplot}[group style={group size=5 by 1, vertical sep= 2.5cm, horizontal sep= 1.5cm}]






\nextgroupplot [	xlabel={\huge $|T|$},	
title= {\huge NYC }]

\addplot[solid, every mark/.append style={solid, fill=white, scale=2}, mark=triangle*]  table[x=i,y=Hakan] {plots/new_NYC_normalized_expected_number_vertices};

\addplot[solid, every mark/.append style={solid, fill=black, scale=2}, mark=otimes*]  table[x=i,y=Max-connection] {plots/new_NYC_normalized_expected_number_vertices};

\addplot[solid, every mark/.append style={solid, fill=white, scale=2}, mark=o]  table[x=i,y=Random] {plots/new_NYC_normalized_expected_number_vertices};	

\addplot[solid, every mark/.append style={solid, fill=black, scale=1.5}, mark=square*]  table[x=i,y=T-Immu] {plots/new_NYC_normalized_expected_number_vertices};

\addplot[solid, every mark/.append style={solid, fill=white, scale=1.5}, mark=oplus*]  table[x=i,y=DIA] {plots/new_NYC_normalized_expected_number_vertices};

\nextgroupplot [	xlabel={\huge $|T|$},	
title= {\huge Tokyo }]

\addplot[solid, every mark/.append style={solid, fill=white, scale=2}, mark=triangle*]  table[x=i,y=Hakan] {plots/TKY_normalized_expected_number_vertices};

\addplot[solid, every mark/.append style={solid, fill=black, scale=2}, mark=otimes*]  table[x=i,y=Max-connection] {plots/TKY_normalized_expected_number_vertices};

\addplot[solid, every mark/.append style={solid, fill=white, scale=2}, mark=o]  table[x=i,y=Random] {plots/TKY_normalized_expected_number_vertices};

\addplot[solid, every mark/.append style={solid, fill=black, scale=1.5}, mark=square*]  table[x=i,y=TImmu] {plots/TKY_normalized_expected_number_vertices};

\addplot[solid, every mark/.append style={solid, fill=white, scale=1.5}, mark=oplus*]  table[x=i,y=DIA] {plots/TKY_normalized_expected_number_vertices};

\nextgroupplot [	xlabel={\huge $|T|$},	
title= {\huge \SG }]

\addplot[solid, every mark/.append style={solid, fill=white, scale=1.5}, mark=triangle*]  table[x=i,y=Hakan] {plots/new_safegraph_normalized_expected_number_vertice};

\addplot[solid, every mark/.append style={solid, fill=black, scale=1.5}, mark=otimes*]  table[x=i,y=Max-connection] {plots/new_safegraph_normalized_expected_number_vertice};

\addplot[solid, every mark/.append style={solid, fill=white, scale=1.5}, mark=o]  table[x=i,y=Random] {plots/new_safegraph_normalized_expected_number_vertice};

\addplot[solid, every mark/.append style={solid, fill=black, scale=1.5}, mark=square*]  table[x=i,y=TImmu] {plots/new_safegraph_normalized_expected_number_vertice};

\addplot[solid, every mark/.append style={solid, fill=white, scale=1.5}, mark=oplus*]  table[x=i,y=DIA] {plots/new_safegraph_normalized_expected_number_vertice};

\nextgroupplot [	xlabel={\huge $|T|$},	
title= {\huge Haslemere }]

\addplot[solid, every mark/.append style={solid, fill=white, scale=2}, mark=triangle*]  table[x=i,y=Hakan] {plots/BBC_normalized_expected_number_vertices};

\addplot[solid, every mark/.append style={solid, fill=black, scale=2}, mark=otimes*]  table[x=i,y=Max-connection] {plots/BBC_normalized_expected_number_vertices};

\addplot[solid, every mark/.append style={solid, fill=white, scale=2}, mark=o]  table[x=i,y=Random] {plots/BBC_normalized_expected_number_vertices};

\addplot[solid, every mark/.append style={solid, fill=black, scale=1.5}, mark=square*]  table[x=i,y=TImmu] {plots/BBC_normalized_expected_number_vertices};

\addplot[solid, every mark/.append style={solid, fill=white, scale=1.5}, mark=oplus*]  table[x=i,y=DIA] {plots/BBC_normalized_expected_number_vertices};

\addplot[solid, every mark/.append style={solid, fill=black, scale=2}, mark=star]  table[x=i,y=greedy-IM] {plots/BBC_normalized_expected_number_vertices};

\nextgroupplot [	xlabel={\huge $|T|$},	
title= {\huge Italy }]

\addplot[solid, every mark/.append style={solid, fill=white, scale=2}, mark=triangle*]  table[x=i,y=Hakan] {plots/Italy_normalized_expected_number_vertices};

\addplot[solid, every mark/.append style={solid, fill=black, scale=2}, mark=otimes*]  table[x=i,y=Max-connection] {plots/Italy_normalized_expected_number_vertices};

\addplot[solid, every mark/.append style={solid, fill=white, scale=2}, mark=o]  table[x=i,y=Random] {plots/Italy_normalized_expected_number_vertices};

\addplot[solid, every mark/.append style={solid, fill=black, scale=1.5}, mark=square*]  table[x=i,y=TImmu] {plots/Italy_normalized_expected_number_vertices};

\addplot[solid, every mark/.append style={solid, fill=white, scale=1.5}, mark=oplus*]  table[x=i,y=DIA] {plots/Italy_normalized_expected_number_vertices};

\addplot[solid, every mark/.append style={solid, fill=black, scale=2}, mark=star]  table[x=i,y=greedy-IM] {plots/Italy_normalized_expected_number_vertices};

	






\end{groupplot}
		
\end{tikzpicture}

%% file: plots/hyperparameter_table.tex
\begin{table}[t]
\caption{The default values of hyper-parameters for different datasets}   
\label{hyper_parameter_datasets}
\begin{tabular}{l|ll}
\hline
Dataset &  Hyper-parameters & $p$ \\ 
\hline
\NYC, \SG      & $a$=0, $b$=0.05, $\rho_2$=0.1   & Using Equation~\ref{definition:initial-probability}\\ 
\Tokyo  &  $a$=0, $b$=0.01 , $\rho_2$=0.1 & Using Equation~\ref{definition:initial-probability}\\ 
\BBC    & $a$=0.05, $b$=0, $\rho_1$=0.1, $l$=5 & Using Equation~\ref{definition:initial-probability}\\ 
\Italy & $\times$  & Using provided $p$ values from data \\
\wiki, \cithep & $\times$ & Using uniform distribution \\
\hline
\end{tabular}
\end{table}

%% file: plots/hyperparams.tex
\begin{figure}[h!]
\centering
\scalebox{0.8}{
\fbox{
$\qed$ \LARGE\emph{RSM}, \ \ \ $\triangle$ \LARGE\emph{ESM}
 }
}\\

{\input{plots/BBC_hyperparameter_fixed_rho1}}  {(a) Different values of $a$ on {\BBC}
}
\vspace{3mm}

{\input{plots/BBC_hyperparameter_fixed_a}  {(b) Different values of $\rho_1$ on {\BBC}}
    		\vspace{3mm}}
{\input{plots/BBC_hyperparameter_distance}  {(c) Different values of $l$ on {\BBC}}
    		}
    \caption{\textit{Spread resulting from different values of $a$, $\rho_1$, $l$ (|T|=3,k=50)}}
    \label{BBC_hyperparameter_performance}   
\end{figure}

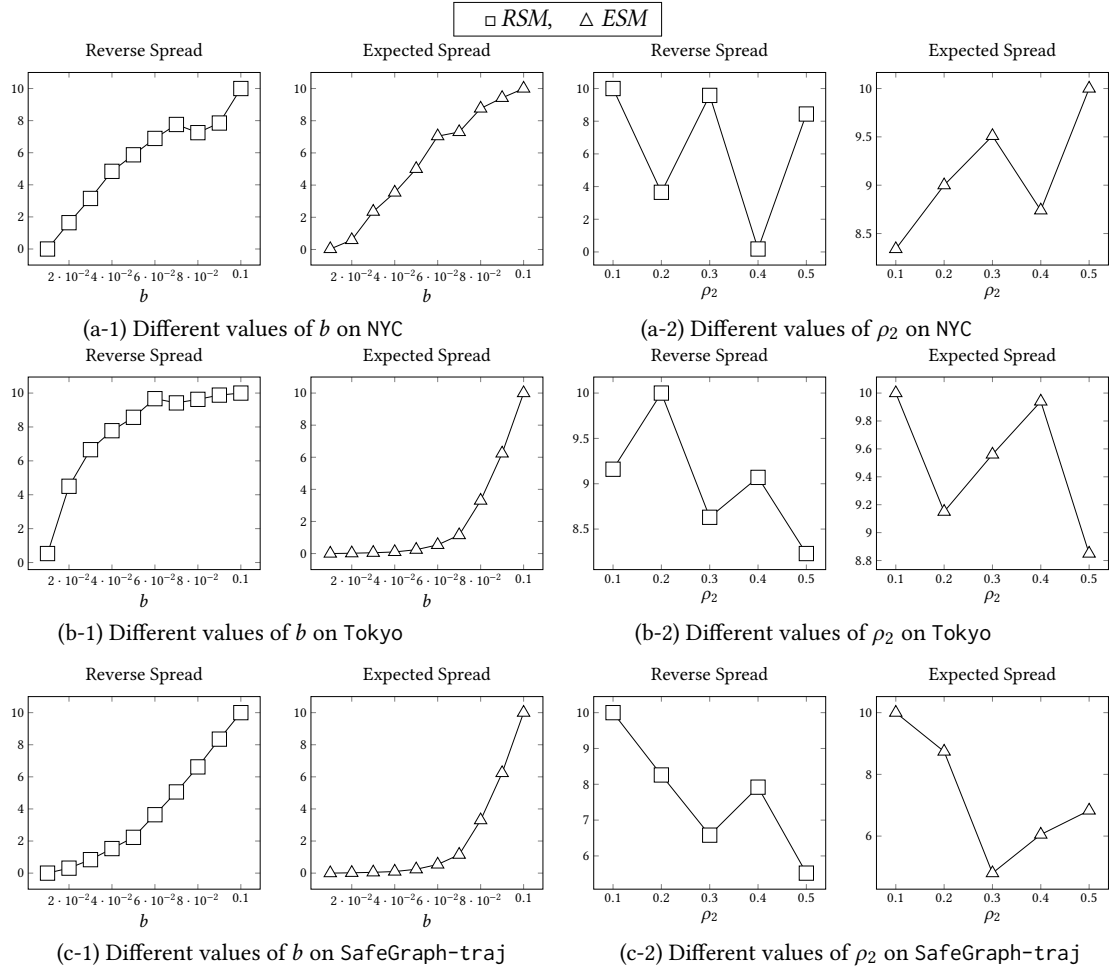
\begin{figure}[h!]
\centering
\scalebox{0.8}{
\fbox{
$\qed$ \LARGE\emph{RSM}, \ \ \ $\triangle$ \LARGE\emph{ESM}
 }
}
{\input{plots/NYC_hyperparameter_fixed_b}  \hspace{-8mm}{\vspace{2mm}(a-1) Different values of $b$ on {\NYC}} \hspace{30mm} {\vspace{-2mm}(a-2) Different values of $\rho_2$ on {\NYC}}}
{\input{plots/TKY_hyperparameter_fixed_b}  \hspace{-8mm}{\vspace{-2mm}(b-1) Different values of $b$ on {\Tokyo}}\hspace{30mm} {\vspace{-2mm}(b-2) Different values of $\rho_2$ on {\Tokyo}}}
\vspace{6mm}

{\input{plots/safegraph_hyperparameter_fixed_b} \hspace{3mm}{\vspace{-2mm}(c-1) Different values of $b$ on {\SG}}\hspace{15mm}{\vspace{-2mm}(c-2) Different values of $\rho_2$ on {\SG}}}
\\
\vspace{3mm}
    \caption{\textit{Spread resulting from different values of $b$ and $\rho_2$ (|T|=25,k=50)}}
    \label{hyperparameter_performance}   
\end{figure}

%% file: plots/BBC_hyperparameter_fixed_rho1.tex
\begin{tikzpicture}[thick,scale=0.6, every node/.style={scale=1.0}]

\begin{groupplot}[group style={group size=2 by 1, vertical sep= 0.2cm, horizontal sep= 2cm}]
\nextgroupplot [xlabel={\huge $a$},	
title= {\huge Reverse Spread }]
\addplot[solid, every mark/.append style={solid, fill=white, scale=3.0}, mark=square*]  table[x=i,y=Gunduz] {plots/BBC_reverse_spread_rho1_0.1};


\nextgroupplot [	xlabel={\huge $a$},	
title= {\huge Expected Spread }]

\addplot[solid, every mark/.append style={solid, fill=white, scale=3.0}, mark=triangle*]  table[x=i,y=Hakan] {plots/BBC_expected_spread_rho1_0.1};

	


\end{groupplot}
\label{safegraph_hyper}		
\end{tikzpicture}

%% file: plots/BBC_hyperparameter_fixed_a.tex
\begin{tikzpicture}[thick,scale=0.6, every node/.style={scale=1.0}]

\begin{groupplot}[group style={group size=5 by 1, vertical sep= 0.2cm, horizontal sep= 2cm}]
\nextgroupplot [xlabel={\huge $\rho_1$},	
title= {\huge Reverse Spread }]

\addplot[solid, every mark/.append style={solid, fill=white, scale=3.0}, mark=triangle*]  table[x=i,y=Hakan] {plots/BBC_reverse_spread_a0.05};

\nextgroupplot [	xlabel={\huge $\rho_1$},	
title= {\huge Expected Spread }]

\addplot[solid, every mark/.append style={solid, fill=white, scale=3.0}, mark=triangle*]  table[x=i,y=Hakan] {plots/BBC_expected_spread_a0.05};

	


\end{groupplot}
\label{safegraph_hyper}		
\end{tikzpicture}

%% file: plots/BBC_hyperparameter_distance.tex
\begin{tikzpicture}[thick,scale=0.6, every node/.style={scale=1.0}]

\begin{groupplot}[group style={group size=2 by 1, vertical sep= 0.2cm, horizontal sep= 2cm}]
\nextgroupplot [xlabel={\huge $l$},	
title= {\huge Reverse Spread }]
\addplot[solid, every mark/.append style={solid, fill=white, scale=3.0}, mark=square*]  table[x=i,y=Gunduz] {plots/BBC_reverse_spread_distance};


\nextgroupplot [	xlabel={\huge $l$},	
title= {\huge Expected Spread }]

\addplot[solid, every mark/.append style={solid, fill=white, scale=3.0}, mark=triangle*]  table[x=i,y=Hakan] {plots/BBC_expected_spread_distance};

	


\end{groupplot}
\label{safegraph_hyper}		
\end{tikzpicture}

%% file: plots/NYC_hyperparameter_fixed_b.tex
\begin{tikzpicture}[thick,scale=0.45, every node/.style={scale=1.1}]


\begin{groupplot}[group style={group size=4 by 1, vertical sep= 1.0cm, horizontal sep= 1.5cm}]

\nextgroupplot [xlabel={\huge $b$},	
title= {\huge Reverse Spread }]
\addplot[solid, every mark/.append style={solid, fill=white, scale=3.0}, mark=square*]  table[x=i,y=Gunduz] {plots/NYC_reverse_spread_rho2_0.1};


\nextgroupplot [	xlabel={\huge $b$},	
title= {\huge Expected Spread }]

\addplot[solid, every mark/.append style={solid, fill=white, scale=3.0}, mark=triangle*]  table[x=i,y=Hakan] {plots/NYC_expected_spread_rho2_0.1};

\nextgroupplot [xlabel={\huge $\rho_2$},	
title= {\huge Reverse Spread }]
\addplot[solid, every mark/.append style={solid, fill=white, scale=3.0}, mark=square*]  table[x=i,y=Gunduz] {plots/NYC_reverse_spread_b0.05};


\nextgroupplot [	xlabel={\huge $\rho_2$},	
title= {\huge Expected Spread }]

\addplot[solid, every mark/.append style={solid, fill=white, scale=3.0}, mark=triangle*]  table[x=i,y=Hakan] {plots/NYC_expected_spread_b0.05};

	


\end{groupplot}
\label{safegraph_hyper}		
\end{tikzpicture}

%% file: plots/TKY_hyperparameter_fixed_b.tex
\begin{tikzpicture}[thick,scale=0.45, every node/.style={scale=1.1}]


\begin{groupplot}[group style={group size=4 by 1, vertical sep= 1.0cm, horizontal sep= 1.5cm}]
\nextgroupplot [xlabel={\huge $b$},	
title= {\huge Reverse Spread }]
\addplot[solid, every mark/.append style={solid, fill=white, scale=3.0}, mark=square*]  table[x=i,y=Gunduz] {plots/TKY_reverse_spread_rho2_0.1};


\nextgroupplot [	xlabel={\huge $b$},	
title= {\huge Expected Spread }]

\addplot[solid, every mark/.append style={solid, fill=white, scale=3.0}, mark=triangle*]  table[x=i,y=Hakan] {plots/safegraph_expected_spread_rho2_0.1};

\nextgroupplot [xlabel={\huge $\rho_2$},	
title= {\huge Reverse Spread }]
\addplot[solid, every mark/.append style={solid, fill=white, scale=3.0}, mark=square*]  table[x=i,y=Gunduz] {plots/TKY_reverse_spread_b0.01};


\nextgroupplot [	xlabel={\huge $\rho_2$},	
title= {\huge Expected Spread }]

\addplot[solid, every mark/.append style={solid, fill=white, scale=3.0}, mark=triangle*]  table[x=i,y=Hakan] {plots/TKY_expected_spread_b0.01};

	


\end{groupplot}
\label{safegraph_hyper}		
\end{tikzpicture}

%% file: plots/safegraph_hyperparameter_fixed_b.tex
\begin{tikzpicture}[thick,scale=0.45, every node/.style={scale=1.1}]


\begin{groupplot}[group style={group size=4 by 1, vertical sep= 1.0cm, horizontal sep= 1.5cm}]
\nextgroupplot [xlabel={\huge $b$},	
title= {\huge Reverse Spread }]
\addplot[solid, every mark/.append style={solid, fill=white, scale=3.0}, mark=square*]  table[x=i,y=Gunduz] {plots/safegraph_reverse_spread_rho2_0.1};


\nextgroupplot [	xlabel={\huge $b$},	
title= {\huge Expected Spread }]

\addplot[solid, every mark/.append style={solid, fill=white, scale=3.0}, mark=triangle*]  table[x=i,y=Hakan] {plots/safegraph_expected_spread_rho2_0.1};

\nextgroupplot [xlabel={\huge $\rho_2$},	
title= {\huge Reverse Spread }]
\addplot[solid, every mark/.append style={solid, fill=white, scale=3.0}, mark=square*]  table[x=i,y=Gunduz] {plots/safegraph_reverse_spread_b0.05};


\nextgroupplot [	xlabel={\huge $\rho_2$},	
title= {\huge Expected Spread }]

\addplot[solid, every mark/.append style={solid, fill=white, scale=3.0}, mark=triangle*]  table[x=i,y=Hakan] {plots/safegraph_expected_spread_b0.05};





\end{groupplot}
\label{safegraph_hyper}		
\end{tikzpicture}

%% file: plots/running_efficiency_hyperNets.tex
\begin{table}[t]
\caption{Running time (in seconds) with different number of nets $|N|$ ($|T|=5$)}   
\label{running_efficiency_against_different_hypernets}
\begin{tabular}{l|rrrrr}
\hline
\diagbox{Dataset}{|N|} & 20000 & 40000  & 60000 & 80000  & 100000 \\ 
\hline
\NYC                                 & 0.99  & 1.95   & 2.96  & 4.00  &4.99  \\ 
\Tokyo                               & 0.43  & 0.86  & 1.31 & 1.75  & 2.18  \\ 
\SG     & {0.04} & {0.07} & {0.11} & {0.15} & {0.19} \\ 
\wiki                           & 53.40   & 108.76   & 158.52   & 216.58   & 271.99  \\ 
\cithep & 1.27  & 2.69   & 3.98  & 5.35   & 6.45  \\ 
\hline
\end{tabular}
\end{table}  

%% file: plots/running_efficiency_different_timeWindows.tex
\begin{table}[t]
\caption{Running time (in seconds) with different number of time windows $|T|$  ($|N|=20K$)} 
\label{running_efficiency_against_different_time_windows}
\begin{tabular}{l|rrrrr}
\hline
\diagbox{Dataset}{|T|} & 5      & 10     & 15      & 20      & 25      \\ 
\hline
\NYC                                 & 0.98   & 1.80  & 3.34   & 5.51   & 11.32    \\ 
\Tokyo                               & 0.42  & 5.40  & 22.65  & 55.94  & 106.0  \\ 
\SG      & {0.04} & {0.06} & {0.10} & {0.12} & {0.16} \\ 
\wiki                          & 53.40   & 149.49  & 320.87   & 549.95  & 917.51 \\ 
\cithep  & 1.27   & 7.46  & 25.68   & 63.73  & 130.65  \\ 
\hline
\end{tabular}
\end{table}

%% file: plots/sampling_based_method_runningtime.tex
\begin{table}[t]
\caption{Running time (in seconds) with different sizes of solution set |S|=k on {\BBC}}   
\label{running_efficiency_samplingbased_on_BBC}
\begin{tabular}{l|rrrrr}
\hline
\diagbox{Method}{k}   & 10 & 20  & 30 & 40  & 50 \\ \hline
RSM                                 & {19.15}  & {19.16}   & {19.20}  & {19.23}  & {19.26}  \\ 
ESM                               & 18.34
  & 18.34  & 18.35 & 18.35  & 18.36  \\ 
Greedy-IM     & 878.98 & 2647.99 & 6729.71 & 11780.86 & 18948.14 \\ \hline
\end{tabular}
\end{table}

%% file: plots/backward_tracing_table.tex
\begin{table}[t]
\caption{Contribution ($\%$ activations among nodes of $S$) of backward traced superspreaders with different sizes of solution set $|S|\!=\!k$}
\label{contact_tracing_contribution}
\begin{tabular}{l|ccccc}
\toprule
\diagbox{Dataset}{k}         & 10   & 20   & 30   & 40   & 50   \\ 
\midrule
\NYC       & 39.9 & 53.6 & 63.3 & 71.1 & 77.3 \\ 
\Tokyo     & 77.8 & 91.0 & 94.3 & 95.5 & 96.1 \\ 
\SG      & 32.0 & 37.2 & 41.5 & 45.3 & 48.9 \\ 
\BBC & 67.9 & 79.1 & 86.0 & 91.1 & 95.0 \\ 
\Italy     & 16.7 & 31.4 & 44.5 & 56.1 & 66.6 \\ 
\bottomrule
\end{tabular}
\end{table}

%% file: 06-conclusion.tex
\section{Conclusion}
\label{section:6}

We introduced the \emph{Temporal Independent Cascade (T-IC)} model for \emph{Reverse Spread Maximization} and \emph{Expected Spread Maximization} tasks and illustrated their application for disease monitoring. We showed that 
reverse spread under the T-IC model is submodular, and proposed efficient algorithms for RSM and ESM with approximation guarantees that handle large-scale, highly granular data. Our novel objectives are to identify i) a minimal set of sentinel nodes (e.g., to minimally cover the network for detection of outbreaks), and ii) a set of highly susceptible nodes (e.g., for prioritizing tracing and treatment). 
Through extensive experiments performed on seven real-world datasets, we showed that RSM significantly outperforms the alternatives for the former task while the ESM solution sets capture significantly more susceptible individuals for the latter. We observed that the dynamic topology captured by our model plays a much more significant role than local connectivity, and that temporal characteristics alongside global graph structure are needed for optimal solutions. 
In particular, ESM is found to dramatically outperform Max-Deg (by up to $82\%$) as a superior targeted strategy for contact tracing, while the sentinel nodes identified by RSM have significantly higher success rates of detecting outbreaks compared to T-Immu and DIA. 
We also presented how the proposed approach can enable more applications by handling individual-level contact information and accounting for temporally ordered events, while being scalable to large networks.
We further applied T-IC to quantify the significant impacts of superspreader venues and events, targeted interventions, and backward contact tracing on contact networks.


